\documentclass[twoside,superscriptaddress,nofootinbib,showpacs]{revtex4}

\usepackage{srcltx}
\usepackage[utf8x]{inputenc}
\usepackage{amssymb}
\usepackage{amsfonts}
\usepackage{amsthm,bbm}
\usepackage[dvips]{color}
\usepackage{epsfig}
\usepackage{fancyhdr}
\usepackage{dsfont}
\usepackage[usenames,dvipsnames]{pstricks}
\usepackage{amsmath}
\newtheorem{thm}{Theorem}[section]
\newtheorem{defi}[thm]{Definition}
\newtheorem{lem}[thm]{Lemma}
\newtheorem{prop}[thm]{Proposition}

\newtheorem{remark}[thm]{Remark}
\newtheorem{example}[thm]{Example}
\def\fin{\hfill $\lozenge$}
\newcommand{\CQCA}{T}
\newcommand{\varia}{u}
\newcommand{\sca}{{\textbf{\emph{t}}}} 
\newcommand{\scb}{{\textbf{\emph{b}}}}

\def\weyl#1{{\bf w}(#1)}
\newcommand{\Id}{\mathds{1}}
\newcommand{\ZZ}{\mathbb{Z}}
\newcommand{\NN}{\mathbb{N}}

\newcommand{\FA}{\mathfrak{A}}
\newcommand{\CN}{\mathcal{N}}
\newcommand{\sollgleich}{\stackrel{!}{=}}
\newcommand{\trans}[1]{\tau_{#1}}
\newcommand{\phase}[1]{\lambda({#1})}
\def\SL#1#2{{\rm SL}(#1,#2)}
\def\subfring{{\mathcal R}}
\def\fring{{\mathcal P}}
\def\isosub{{\mathcal I}}
\def\ftra{{\left(\begin{array}{cc}0&1\\1&0\end{array}\right)}}
\def\htra{{\left(\begin{array}{cc}1&0\\1&1\end{array}\right)}}
\def\shear#1{{\left(\begin{array}{cc}1&0\\{#1}&1\end{array}\right)}}
\def\genp{P}

\newcommand{\elgen}{G}
\newcommand{\tr}{\mathrm{tr}\,}
\newcommand{\dg}[1]{\mathrm{dg}(#1)}
\newcommand{\glider}{G}
\newcommand{\fractal}{F}
\newcommand{\glidersca}{{\textbf{\emph{g}}}}
\newcommand{\fractalsca}{{\textbf{\emph{f}}}}
\newcommand{\simpleglider}{G_s}
\newcommand{\simpleglidersca}{\textbf{\emph{g}}_s}
\newcommand{\stab}{\mathcal{S}}
\newcommand{\ul}[1]{\underline{#1}}
\newcommand{\ket}[1]{\left|#1\right>}
\graphicspath{{graphics_eps/}}
\begin{document}
\title{Time Asymptotics and Entanglement Generation of Clifford Quantum Cellular Automata}
\author{Johannes G\"utschow}
\email{johannes.guetschow(at)itp.uni-hannover.de}
\affiliation{Institut f\"ur
Mathematische Physik, Technische Universit\"at Braunschweig,
Mendelssohnstra{\ss}e~3, 38106 Braunschweig, Germany}
\affiliation{Institut f\"ur Theoretische Physik, Universit\"at Hannover, Appelstra{\ss}e 2, 30167 Hannover}
\author{Sonja Uphoff}
\email{s.uphoff(at)tu-bs.de}
\affiliation{Institut f\"ur
Mathematische Physik, Technische Universit\"at Braunschweig,
Mendelssohnstra{\ss}e~3, 38106 Braunschweig, Germany}
\affiliation{Institut f\"ur Rechnergest\"utzte Modellierung im Bauingenieurwesen, Technische Universit\"at Braunschweig, Pockelsstra{\ss}e 3, 38106 Braunschweig, Germany}
\author{Reinhard F. Werner}
\email{reinhard.werner(at)itp.uni-hannover.de}
\affiliation{Institut f\"ur
Mathematische Physik, Technische Universit\"at Braunschweig,
Mendelssohnstra{\ss}e~3, 38106 Braunschweig, Germany}
\affiliation{Institut f\"ur Theoretische Physik, Universit\"at Hannover, Appelstra{\ss}e 2, 30167 Hannover}
\author{Zolt\'an Zimbor\'as}
\email{zimboras(at)lusi.uni-sb.de}
\affiliation{Institut f\"ur Theoretische Physik Universit\"at des Saarlandes, 66041 Saarbr\"ucken, Germany}
\affiliation{ISI Foundation, Quantum Information Theory Unit, Viale S. Severo 65, 10133 Torino, Italy}
\date{\today}
%%%
\begin{abstract}
We consider Clifford Quantum Cellular Automata (CQCAs) and their time evolution. CQCAs are an especially simple type of Quantum Cellular Automata, yet they show complex asymptotics and can even be a basic ingredient for universal quantum computation. In this work we study the time evolution of different classes of CQCAs. We distinguish between periodic CQCAs, fractal CQCAs and CQCAs with gliders. We then identify invariant states and study convergence properties of classes of states, like quasifree and stabilizer states. Finally we consider the generation of entanglement analytically and numerically for stabilizer and quasifree states.
\end{abstract}
%%%
%%%
\maketitle
\tableofcontents
%%%
\pagestyle{fancy}
%%%
%%%
\section{Introduction}\label{sec:intro}
%%%
%%%
Quantum cellular automata (QCAs), i.e., reversible quantum systems
which are discrete both in time and in space \cite{Werner2004}, and
exhibit strictly finite propagation, have recently come under study
from different directions. On the one hand, they serve as one of the
computational paradigms for quantum computation, and it has been
shown that a certain one-dimensional QCA with twelve states per cell
can efficiently simulate all quantum computers \cite{Shepherd2006}. On the practical
side, QCAs are a direct axiomatization of the kind of quantum
simulator in optical lattices, which are under construction in
many labs at the moment \cite{Bloch,Meschede}. Related to this,
they can be seen as a paradigm of quantum lattice systems, in which
the consequences of locality, assumed in the idealized pure form of
strictly finite propagation, can be explored directly. Due to the
famous Lieb-Robinson bounds \cite{Lieb1972,Nachtergaele,Hastings2006,Eisert2006} this feature
is also present in continuous time models albeit in an approximate
form.

In all these settings, the time asymptotics for the iteration of the
QCA is of interest, and displays a curious dichotomy between a global
and a local point of view. On the one hand we are assuming
reversibility, so the global evolution is an automorphism taking pure
states to pure states. If we split the system into two subsystems,
e.g., a right half chain and a left half chain, then we expect the
QCA to generate entanglement from any initial product state. There is
a simple upper bound (see Section \ref{sec:bounds}) showing that the entanglement
growth is at most linear, and we find indeed that for the automata
studied in this paper this is the typical behavior. However,
utilizing this entanglement requires the control of larger and larger
regions. So from a local perspective, i.e., when we only consider the
restriction of the state to a finite region, we will not see this
increase. In fact, a typical behavior of the local restrictions is
the convergence to the maximally mixed state. That for large times
the state seems globally pure and locally completely mixed is no
contradiction: it merely reflects the fact that any local system
becomes maximally entangled with its environment.

Stationary states are in some sense the final result of an asymptotic
evolution. Again the local analysis makes it clear that the totally
mixed state is invariant for any reversible cellular automaton. In
general there may be many more invariant states, among them some,
which are not only ergodic (i.e., extremal in the set of all
translation invariant states), but even pure (i.e., extremal in
the set of all states). For a special class we consider here, we
exhibit a rich set of such states.

The very fact that QCAs can serve as a universal computational model
suggests that asymptotic questions cannot be easily answered in full
generality. Therefore, in this paper we look at a subclass of QCAs,
the Clifford Quantum Cellular Automata (CQCAs). In this case much of
the essential information can be obtained by studying a classical
cellular automaton, which even turns out to be of a linear type.
Consequently, we can answer some questions exhaustively. The drawback
is, of course, that universal computation is not possible in this
class. However, we believe that some typical features of QCA asymptotics
can be studied in this theoretical laboratory .

Our paper is organized as follows. We begin with a short introduction to CQCAs and their representation as $2\times 2$-matrices in Section \ref{sec:cqca_intro}. Then we derive a general theory of the time evolution of CQCAs in Section \ref{sec:characterization}. We show that CQCAs can be divided in three major classes by their time evolution: periodic automata (Section \ref{sec:periodic}), automata which act as lattice translations on special observables we call gliders (Section \ref{sec:gliders}), and fractal automata (Section \ref{sec:frac}), whose space-time picture shows self-similarity on large scales. The class of a CQCA is determined by the trace of its matrix. A constant trace means the automaton is periodic, a trace of the form $u^{-n}+u^{n}$ indicates gliders that move $n$ steps on the lattice each time step. All other CQCAs show fractal behavior. We prove that all automata with gliders and $n=1$ are equivalent in the sense that they can be transformed into each other by conjugation with other CQCAs. This turns out not to be true for glider CQCAs with $n>1$. 

Using the results for the observable asymptotics, in Section \ref{sec:stationary} convergence and invariance of translation-invariant states are analyzed. For periodic CQCAs the construction of mixed invariant states is straightforward (Section \ref{sec:inv_periodic}). There are also pure stabilizer states that are invariant with respect to periodic CQCAs (Section \ref{sec:stab_inv}). Non-periodic CQCAs can not leave stabilizer states invariant. In fact, for fractal CQCAs the only known invariant state is the tracial state. For glider CQCAs invariant states exist, and are constructed as the limit of the time evolution of initial product states in Section \ref{sec:stat_conv_prod}. Finally we consider $n=1$ glider automata and states which are quasifree on the canonical anticommutation relations (CAR) algebra (Section \ref{sec:stat_conv_qf}). We employ the Araki-Jordan-Wigner transformation (Section \ref{sec:ajw}) to transfer the glider CQCA to the CAR algebra and study its action as a Bogoliubov transformation on quasifree states; we study both convergence (Section \ref{sec:conv_qf}) and invariance (Section \ref{sec:stat_qf}).

In the last Section (Section \ref{sec:entanglement}) the entanglement generation of CQCAs is considered. First we derive a general linear upper bound for entanglement generation of QCAs. We then prove that in the translation-invariant case the bound is more restrictive and can be saturated by CQCAs acting on initially pure stabilizer states. For CQCAs acting on stabilizer states the asymptotic entanglement generation rate is governed by the order of the trace polynomial (Section \ref{sec:ent_stab}). We then examine a set of quasifree states interpolating between a product state and a state which is invariant with respect to the standard glider CQCA (Section \ref{sec:ent_qf}). It is shown numerically, that the entanglement generation is linear also in this case, but with a slope that can be arbitrarily small.
%%%
%%%
%%%
\section{Time evolution of the observable algebra}\label{sec:asymptotics}
  CQCAs show a variety of different time evolutions. In this section we develop criteria to predict the time evolution by characteristics of the matrix of the corresponding symplectic cellular automaton (SCA).
\subsection{A short introduction to CQCAs}
\label{sec:cqca_intro}
Clifford Quantum Cellular Automata are a special class of Quantum Cellular Automata first described in \cite{SchlingemannCQCA}. As the name indicates, they are QCAs that use the Clifford group operations. They can be defined for arbitrary lattice dimensions and prime cell dimensions. Here we will only consider the case of a one dimensional infinite lattice and cell dimension two (qubits). Thus we deal with an ordinary spin-chain. QCAs are translation-invariant operations which we define on a quasi-local observable algebra \cite{Werner2004}. Reversible QCAs are automorphisms. So in our case CQCAs are translation-invariant automorphisms of the spin-$\frac{1}{2}$ chain observable algebra $\FA=\otimes_{i=-\infty}^{\infty}\FA_i$, where $\FA_i\cong M_2$. In explicit this means that a CQCA $\CQCA$ commutes with the lattice translation $\tau$ and leaves the product structure of observables invariant:
\begin{equation*}
  \CQCA(AB)=\CQCA(A)\CQCA(B),\quad\forall A,B\in\FA
\end{equation*}
As shown in \cite{Werner2004} a QCA is fully specified by its local transition rule $\CQCA_0$, i.e., the picture of the one-site observable algebra. In our case the Pauli matrices form a basis of this algebra, so their pictures specify the whole CQCA. There is one important restriction on the set of possible images. They have to fulfill the commutation relations of the original Pauli matrices on the same site and also on all other sites. As a QCA in general has a propagation on the lattice, one cell observables are mapped to observables on a neighborhood $\CN$ of the site. Thus neighboring one-site observables may overlap after one time step but still have to commute. This imposes conditions on the local rule. CQCAs are defined as follows:
\begin{defi}
  A Clifford Quantum Cellular Automaton $\CQCA$ is an automorphism of the quasi-local observable algebra of the infinite spin chain that maps tensor products of Pauli matrices to multiples of tensor products of Pauli matrices and commutes with the lattice translation $\tau$. 
\end{defi}

We now want to find a classical description of the CQCA. It is well known that Clifford operations can be simulated efficiently by a classical computer. Therefore it is not surprising that an efficient classical description of CQCAs exists. This description was introduced in \cite{SchlingemannCQCA}. We will only give a short overview of the topic, for proofs and details we refer to the literature. 

We use the (finite) tensor products of Pauli matrices as a basis of the observable algebra. Every CQCA $\CQCA$ maps those tensor products to multiples of tensor products of Pauli matrices. The factor can only be a complex phase, that can be fixed uniquely by the phase for single cell observables (single Pauli matrices). We can thus describe the action of the CQCA $\CQCA$ on Pauli matrices by a classical cellular automaton $\sca$ acting on their labels ``$1,2,3$''. We could keep track of the phase separately, but for our analysis this is unnecessary.

Mathematically we describe this correspondence as follows: The Pauli matrices correspond to a Weyl system over a discrete phase space where all operations are carried out modulus two. For one site we have
\begin{equation*}
  \sigma_1=X=\weyl{1,0},\quad \sigma_2=Y=i\weyl{1,1},\quad \sigma_3=Z=\weyl{0,1},\quad \sigma_0=\Id=\weyl{0,0}.
\end{equation*}
Tensor products of these operators are constructed via
\begin{equation*}
  \weyl{\xi}=\bigotimes_{x\in\ZZ}\weyl{\xi(x)}
\end{equation*}
where $\xi=(\xi_+,\xi_-)$ is a tuple of two binary strings which differ from $0$ on only finitely many places and $\xi(x)$ is its value at position $x$, e.g.\ $\tbinom{1}{0}$. Thus the tensor product is well defined. %, because $\weyl{\xi(x)}$ differs from $\Id$ only on finitely many places. 
Before we continue with the mathematical definition, we want to illustrate the classical description by a simple example: 

\begin{example}\rm
  \label{exam:cqca}
  We define our CQCA on the observable algebra $\FA$ of a spin chain by the rule
  \begin{equation*}
    \begin{array}{rcc}
      &\CQCA_i:& \FA_i\to \FA_{\CN+i},\qquad\qquad\\
      \CQCA\sigma_1^i &=& \sigma_3^i,\\
      \CQCA\sigma_3^i &=& \sigma_3^{i-1}\otimes\sigma_1^i\otimes\sigma_3^{i+1}.
    \end{array}
  \end{equation*}
  The image of $\sigma_2$ follows from the product of the images of $\sigma_1$ and $\sigma_3$:
  \begin{equation*}
    \CQCA\sigma_2^i= -\sigma_3^{i-1}\otimes\sigma_2^i\otimes\sigma_3^{i+1}.
  \end{equation*}
  $\CQCA$ has to be an automorphism to be a CQCA. To verify this we check if the commutation relations are preserved: 
  \begin{eqnarray*}
    \phantom{b}[\CQCA\sigma_1^i,\CQCA\sigma_1^j]&=&[\sigma_3^i,\sigma_3^j]=0,\\
    \phantom{b}[\CQCA\sigma_3^i,\CQCA\sigma_3^j]&=&[\sigma_3^{i-1}\otimes\sigma_1^i\otimes\sigma_3^{i+1},\sigma_3^{j-1}\otimes\sigma_1^j\otimes\sigma_3^{j+1}]=0,        
  \end{eqnarray*}
  and
  \begin{eqnarray*}
    &[\CQCA\sigma_3^i,\CQCA\sigma_1^j]=[\sigma_3^{i-1}\otimes\sigma_1^i\otimes\sigma_3^{i+1},\sigma_3^j]=
      0&i\ne j,\\
    &\{\CQCA\sigma_3^i,T\sigma_1^j\}=
      \{\sigma_3^{i-1}\otimes\sigma_1^i\otimes\sigma_3^{i+1},\sigma_3^j\}=
    0&i= j\; .  
  \end{eqnarray*}
  This automaton will be used extensively in the following parts of the paper, so we give it the name $\simpleglider$.
  If we think of the CQCA $\simpleglider$ as a classical automaton acting on the labels of the Pauli matrices we can illustrate the evolution (for one time step) of the observable $\sigma_3^{-1}\otimes\sigma_2^0\otimes\sigma_1^1$ as follows (the underlined labels are situated at the origin):
  \begin{equation*}
    \simpleglidersca(3\,\ul{2}\,1)=\simpleglidersca(3\,\ul{0}\,0)\cdot \simpleglidersca(0\,\ul{2}\,0)\cdot \simpleglidersca(0\,\ul{0}\,1)=
    \begin{array}{ccccc}
      &3&1&\ul{3}&\\
      \odot&&3&\ul{-2}&3\\
      \odot&&&&3\\
      \hline
      =&3&-i2&\ul{i1}&0
    \end{array}
    =(3\,2\,\ul{1}).
  \end{equation*}
  We observe, that the observable only moves on the lattice under the action of the CQCA $\CQCA$. We call observables with this property gliders. Their existence can be observed easily, when we consider the space-time images of one-cell observables. $\sigma_1$ and $\sigma_3$ generate ``checkerboards'' of $\sigma_1$ and $\sigma_3$ matrices. As $\sigma_1$ is mapped to $\sigma_3$ in the first step, the $\sigma_1$-checkerboard is the same as the $\sigma_3$-checkerboard shifted one step in time. If we additionally shift in the space direction by one cell, the checkerboards are exactly the same up to two diagonals and thus cancel out as shown in Figure \ref{fig:glider}. We thus produced a very simple observable on which the automaton acts as a translation, a basic glider. 
  \begin{figure}[htbp]
    \begin{center}
      \input{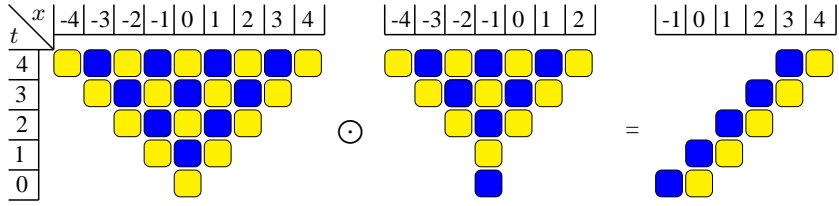}
      \caption{Glider of the example CQCA (\ref{eq:simpleglider}). The blue boxes represent $\sigma_1$, the yellow ones $\sigma_3$.}
      \label{fig:glider}
    \end{center}
  \end{figure}
  
  Another interesting property of this automaton is the fact that it maps the ``all spins up'' product state to a one-dimensional cluster-state, which is a one-dimensional version of the two-dimensional resource-state for the ``One Way Quantum Computing'' scheme by Raussendorf and Briegel \cite{Raussendorf2001}. It is also the basic ingredient of a scheme of ``quantum computation via translation-invariant operations on a chain of qubits'' by Raussendorf \cite{Raussendorf2005a}. In a similar way, the update rule $G_s$ (but with $\sigma_1$ and $\sigma_3$ exchanged) has appeared as time-evolution of spin chains implemented by a Hamiltonian that is subjected to periodic quenches \cite{fitzsimons2006, eisler}, and it has even been realized
experimentally in an NMR-System \cite{fitzsimons2007}. 
   
  In the phase space picture we can describe the automaton by a $2\times 2$-matrix with polynomial entries. In phase space our CQCA-rule reads
  \begin{eqnarray*}
    \simpleglidersca\binom{1}{0}&=&\left(\begin{array}{ccc}0&0&0\\0&1&0\end{array}\right),\\
    \simpleglidersca\binom{0}{1}&=&\left(\begin{array}{ccc}0&1&0\\1&0&1\end{array}\right),\\
    \simpleglidersca\binom{1}{1}&=&\left(\begin{array}{ccc}0&1&0\\1&1&1\end{array}\right).
  \end{eqnarray*}
  Now we transform the binary strings to Laurent polynomials by indicating the position by a multiplication with a variable $\varia$ and add all terms from the different positions to a Laurent polynomial:\footnote{More formally, we perform an algebraic Fourier transformation \cite{SchlingemannCQCA} mapping the binary string $\xi(x)$ to the Laurent polynomial $\hat{\xi}(u)=\sum_{x\in\mathbb{Z}} \xi(x)u^x$.}
  \begin{eqnarray*}
    \simpleglidersca\binom{1}{0}&=&\binom{0}{1},\\
    \simpleglidersca\binom{0}{1}&=&\binom{1}{u^{-1}+u},\\
    \simpleglidersca\binom{1}{1}&=&\binom{1}{u^{-1}+1+u}.
  \end{eqnarray*}
  We arrange the images of $\tbinom{1}{0}$ and $\tbinom{0}{1}$ in a $2\times 2$-matrix
  \begin{equation}
    \label{eq:simpleglider}
    \simpleglidersca=\left(\begin{array}{cc}
        0&1\\
        1&\varia^{-1}+\varia
      \end{array}
    \right).
  \end{equation}
  The image under $\simpleglider$ of an arbitrary tensor product of Pauli matrices is now determined by the multiplication of the corresponding vector of polynomials by the matrix representation $\simpleglidersca$ of $\simpleglider$. We will later argue that this works for all CQCA. 
  \fin
\end{example}

Now we come back to the mathematical definition of CQCAs: 
The Weyl operators fulfill the relation 
\begin{equation*}
  \weyl{\xi + \eta}=(-1)^{\eta_+\xi_-}\weyl{\xi}\weyl{\eta}
\end{equation*}
and therefore the commutation relation 
\begin{equation*}
  \weyl{\xi}\weyl{\eta}=(-1)^{\xi_+\eta_--\xi_-\eta_+}\weyl{\eta}\weyl{\xi}.
\end{equation*}
In both cases terms of the type $\xi_+\eta_-$ are scalar products where the addition is carried out modulo~$2$.
The arguments of the Weyl operators are elements of a vector space over the finite field $\ZZ_2$ which we call phase space and thus commute. Of course the corresponding Weyl operators do not necessarily commute, but they always commute or anticommute. Their commutation relations are encoded in the symplectic form $\sigma(\xi,\eta)=\beta(\xi,\eta)-\beta(\eta,\xi)=\xi_+\eta_--\xi_-\eta_+\in\ZZ_2$. As an automorphism the CQCA leaves the commutation relations invariant. A representation of the CQCA on phase space therefore has to leave the symplectic form invariant. Such a translation-invariant symplectic map is called symplectic cellular automaton (SCA). We can find a SCA and an appropriate phase function $\phase{\xi}$ for every CQCA.
\begin{prop}[\cite{SchlingemannCQCA}]
  \label{prop:phaseCQCA}
  Let $\CQCA$ be a CQCA. Then we can write 
  \begin{equation}
    \CQCA[\weyl{\xi}]=\phase{\xi}\weyl{\sca\xi}
  \end{equation}
  with a symplectic cellular automaton $\sca$ and a translation invariant phase function $\phase{\xi}$ which fulfills
  \begin{equation*}
    \phase{\xi+\eta}=\phase{\xi}\phase{\eta}(-1)^{\beta(\xi,\eta)-\beta(\sca\xi,\sca\eta)}
  \end{equation*}
  as well as $|\phase{\xi}|=1\;\forall\xi$. Furthermore, $\lambda(\xi)$ is uniquely determined for all $\xi$ by the choice of $\lambda$ on one site. 
\end{prop}
In the following analysis of CQCAs we neglect the global phase and consider the symplectic cellular automata only. As we can always find appropriate phase functions all results for SCAs translate to the world of CQCAs directly. 
We have already seen in Example \ref{exam:cqca}, that there exists a very convenient representation of CQCAs as $2\times 2$-matrices with polynomial entries.
\begin{defi}
  $\fring$ is the ring of Laurent polynomials over $\ZZ_2$. $\subfring$ is the subring of $\fring$, which consists of all polynomials, which are reflection invariant with center $u=0$. 
\end{defi}
\begin{thm}
  \label{thm:cqca_prop}
  Every CQCA $\CQCA$ is represented up to a phase by a unique $2\times 2$-matrix $\sca$ with entries from $\fring$. Such a matrix represents a CQCA if and only if
  \begin{itemize}
    \item $\det(\sca)=\varia^{2a},\,a\in\ZZ$; 
    \item all entries are symmetric polynomials centered around the same (but arbitrary) lattice point $a$;
    \item the entries of both column vectors, which are the pictures of $(1,0)$ and $(0,1)$, are coprime.
  \end{itemize}
  \begin{proof}
    We will only give a sketch of the proof. For details see \cite{SchlingemannCQCA}. The connection between CQCAs and SCAs was already established in Proposition \ref{prop:phaseCQCA}. What remains to show is that SCAs are linear transformations over $\fring^2$. The application of a SCA to a vector in phase space can be described as the multiplication of this vector with a matrix representing the SCA from the left. The product is the convolution of binary strings. By the algebraic Fourier transform $\hat f(\varia)=\sum_{x\in\ZZ}f(x)u^x$, which maps the vectors of binary strings onto vectors with entries from the ring $\fring$ of Laurent polynomials over the finite field $\ZZ_2$, the convolution of strings translates into the multiplication of polynomials. Thus in this picture the application of the SCA to a phase space vector is just a common matrix multiplication, where all operations are carried out modulo $2$. If we translate the symplectic form and the condition that it has to be left invariant to the polynomial picture we retrieve the above conditions on the matrix $\sca$.
  \end{proof}  
\end{thm}
We can further simplify these statements by only considering automata centered around $0$. The lattice translation $\tau$ is a SCA which by definition commutes with all other SCAs. Its determinant is given by $\det(\tau)=\varia^2$. Therefore every SCA can be written as the product of a lattice translation and an automaton centered around $0$ which has determinant one. We call these automata centered symplectic cellular automata (CSCAs) and in the following sections we point our focus to those. 

CSCAs and CQCAs each form a group. This group is generated by a countably infinite set of basis automata. 
The CSCA form the group $\Gamma=\SL{2}{\subfring}$, which is the group of all $2\times 2$-matrices with determinant $1$ over the ring of centered Laurent polynomials with binary coefficients. The group $\Gamma_0$ is the group of local automata. Their generators are 
\begin{equation*}
  \label{eq:localfourier}
  H=\htra,\quad \genp=\ftra.
\end{equation*}
Additionally we have the shear transformations
\begin{equation*}
  \elgen_n:=\shear{\varia^n+\varia^{-n}},\,n\in\NN.
\end{equation*}
which complete the set of generators of $\Gamma$. For proofs see \cite{SchlingemannCQCA}.
%%%
\subsection{Classification of Clifford quantum cellular automata}
%%%
\label{sec:characterization}
The time evolution of a CQCA is determined up to a phase by the powers of the matrix of the corresponding CSCA. We will only consider the evolution of single cell observables, as any other observable can be represented by products and sums of these. This product structure is invariant under the action of the automaton, because it is an automorphism of the observable algebra. This means that when we discuss the time evolution of a CQCA $\CQCA$, we will consider the action of powers $\sca^n$ of the matrix of the CSCA on phase space vectors which only contain constants. For example the image of $\sigma_1=\weyl{1,0}$ after $n$ time steps of $\CQCA$ is given by the first column vector of $\sca^n$ (and a global phase). The matrix $\sca$ does not always have an eigenvalue, because it is a matrix over a ring without multiplicative inverses for all elements. But for some of the automata eigenvalues do exist. These automata are called glider-automata, because on a special set of observables, the gliders, they act as lattice translations. 

We will prove that if the trace is a polynomial consisting of only two summands, i.e., it is of the form $\tr \sca=\varia^{-n}+\varia^{n}$, two eigenvectors exist and the automaton has gliders. If this does not hold, we can distinguish two cases. The trace can be either a constant, or an arbitrary symmetric polynomial. In the first case the automaton is periodic, in the second case it generates a time evolution which has fractal properties. As the case of periodic automata is not very interesting and the case of fractal automata will be covered in another paper \cite{Nesme}, we focus on automata with gliders. We prove that all automata with gliders which move one step on the lattice at every timestep are equivalent. We also give an example to show that this is not the case for gliders that move more than one step.
%%%
\subsubsection{Automata with gliders}
%%%
\label{sec:gliders}
We will first define our notion of a glider. Here we consider the case of qubits only, but with minor alterations all results of this section also hold for qudits with prime dimension. This extension and also all proofs which are omitted here can be found in \cite{Guetschow2008a,Uphoff2008}.
\begin{defi}
  A glider is an observable, on which the CQCA acts as a lattice translation. In the Laurent polynomial picture a translation is a multiplication by $\varia^k,\,k\in\ZZ$.
\end{defi}
We have already seen this behavior in Example \ref{exam:cqca}.
Now we determine the conditions a CQCA has to fulfill to have gliders. In general we can not diagonalize the matrices of the corresponding CSCA, because all our calculations are over a finite field and the entries are only polynomials in $\varia$ which are centered palindromes. The polynomials are usually not invertible, so the equations which occur in the diagonalization cannot be solved mechanically. Furthermore, the diagonal matrix would not correspond to a CSCA as $\varia^{n}$ and $\varia^{-n}$ are not centered palindromes. Hence we take a different approach.
First, let us introduce some terms: We call a glider $\xi=(\xi_+,\xi_-)$ a \textit{minimal glider} iff its two entries in phase space $\xi_+$, $\xi_-$ have no common non-invertible divisor\footnote{The phase space vector of a minimal glider is maximal with respect to the notation introduced in \cite{SchlingemannCQCA}.}. The \textit{wedge-product} of two phase space vectors shall be defined as $\xi \wedge \eta=\xi_+ \eta_- -\eta_+ \xi_-$. As we deal with qubits here, addition and multiplication of polynomials are carried out modulo two and $\xi \wedge \eta=\xi_+ \eta_- +\eta_+ \xi_-$.
We define the \textit{involution} of a polynomial $p$ as the substitution of $u$ by $u^{-1}$ and denote it by $\bar{p}$.
Finally we will also need the following proposition:

\begin{prop}[\cite{Guetschow2008a,Uphoff2008}]\label{mono}
In the ring $\fring$ of Laurent polynomials over the finite field $\ZZ_2$, the only invertible elements are monomials. 
\end{prop}

Now we have all we need to state following theorems:
\begin{prop}
  \label{prop:glider_sca_props}
  Given a CSCA $\sca$ and a non-zero phase space vector $\xi$ with $\sca\xi=\varia^n\xi$, $n\in\NN$ the following is true:
  \begin{enumerate}
    \item $\bar\xi$ fulfills $\sca\bar\xi=\varia^{-n}\bar\xi$, thus it is a glider with the same speed but different direction as $\xi$.
    \item $\sca$ is uniquely given by
    \begin{eqnarray}
      \sca_{11} & = & \frac{\varia^n\xi_+\bar\xi_-+\varia^{-n}\bar\xi_+\xi_-}{\xi\wedge\bar\xi},\label{eq:auto_aus_laeuf_1}\\
      \sca_{12} & = & \frac{\left(\varia^n+\varia^{-n}\right)\,\xi_+\bar\xi_+}{\xi\wedge\bar\xi},\\
      \sca_{21} & = & \frac{\left(\varia^n+\varia^{-n}\right)\,\xi_-\bar\xi_-}{\xi\wedge\bar\xi},\\
      \sca_{22} & = & \frac{\varia^{-n}\xi_+\bar\xi_-+\varia^n\bar\xi_+\xi_-}{\xi\wedge\bar\xi}.\label{eq:auto_aus_laeuf_4}
    \end{eqnarray}
    \item $tr (\sca)=\varia^{-n}+\varia^{n}$
    \item All gliders are multiples of
    \begin{equation}
      \label{eq:min_laeuf}
      \xi=\binom{\xi_+}{\xi_-}=\binom{\frac{\sca_{12}}{\mathrm{gcd}\left(\varia^{n}+\sca_{11},\sca_{12}\right)}}{\frac{\varia^{n}+\sca_{11}}{\mathrm{gcd}\left(\varia^{n}+\sca_{11},\sca_{12}\right)}},
    \end{equation}
    or
    \begin{equation}
      \label{eq:min_laeuf_bar}
      \bar\xi=\binom{\bar\xi_+}{\bar\xi_-}=\binom{\frac{\sca_{12}}{\mathrm{gcd}\left(\varia^{-n}+\sca_{11},\sca_{12}\right)}}{\frac{\varia^{-n}+\sca_{11}}{\mathrm{gcd}\left(\varia^{-n}+\sca_{11},\sca_{12}\right)}}.
    \end{equation}
  \end{enumerate}
  \begin{proof}
    \begin{enumerate}
      \item We use $\sca\xi=\varia^n\xi$ and take the involution on both sides. $\sca$ consists of palindromes, thus $\bar\sca=\sca$ and we get $\sca\bar\xi=\varia^{-n}\bar\xi$.
      \item We write $\sca\xi=\varia^n\xi$ and $\sca\bar\xi=\varia^{-n}\bar\xi$ component wise yielding the four equations
      \begin{displaymath}
        \begin{array}{crcl}
          (\mathrm{I}) & \sca_{11}\xi_++\sca_{12}\xi_- & = & \varia^n\xi_+, \\
          (\mathrm{II}) & \sca_{21}\xi_++\sca_{22}\xi_- & = & \varia^n\xi_-, \\
          (\bar{\mathrm{I}}) & \sca_{11}\bar\xi_++\sca_{12}\bar\xi_- & = & \varia^{-n}\bar\xi_+, \\
          (\bar{\mathrm{II}}) & \sca_{21}\bar\xi_++\sca_{22}\bar\xi_- & = & \varia^{-n}\bar\xi_-. 
        \end{array}
      \end{displaymath}
      Combining them in the right way we get
      \begin{eqnarray}
        \sca_{11}\left(\xi\wedge\bar\xi\right) & = & \varia^n\xi_+\bar\xi_-+\varia^{-n}\bar\xi_+\xi_-,\label{eq:a11}\\
        \sca_{12}\left(\xi\wedge\bar\xi\right) & = & \left(\varia^n+\varia^{-n}\right)\xi_+\bar\xi_+,\label{eq:a12}\\
        \sca_{21}\left(\xi\wedge\bar\xi\right) & = & \left(\varia^n+\varia^{-n}\right)\xi_-\bar\xi_-,\label{eq:a21}\\
        \sca_{22}\left(\xi\wedge\bar\xi\right) & = & \varia^{-n}\xi_+\bar\xi_-+\varia^n\bar\xi_+\xi_-. \label{eq:a22}
      \end{eqnarray}
      By assumption $\sca$ is a CSCA, so the division by $\xi\wedge\bar\xi$ gives a polynomial result and we have Equations (\ref{eq:auto_aus_laeuf_1}) to (\ref{eq:auto_aus_laeuf_4}). In Proposition \ref{gliderCQCA} we will show which conditions $\xi$ has to fulfill for the division to be valid and therefore $\xi$ to be a glider.
      \item 
      \begin{displaymath}
        \tr\sca=\sca_{11}+\sca_{22}=\frac{\left(\xi\wedge\bar\xi\right)\varia^{-n}+\left(\xi\wedge\bar\xi\right)\varia^{n}}{\left(\xi\wedge\bar\xi\right)}=\varia^{-n}+\varia^{n}
      \end{displaymath}
      \item We now use $(\mathrm{I})$ and $(\mathrm{II})$ together with $\det \sca=1$ and $\tr \sca=\varia^{-n}+\varia^{n}$ to derive the form of $\xi$. We get the equation 
      \begin{displaymath}
        \xi_+\left(\varia^{n}-\sca_{11}\right) = \xi_-\sca_{12}.
      \end{displaymath}
      This equation for $\xi_+$ and $\xi_-$, has still one free parameter. One particular solution for the equation is $\xi_+^{(part)}=\sca_{12}$, $\xi_{-}^{(part)}=\varia^n-\sca_{11}$. To obtain the minimal glider we have to divide these components of the particular solution by their greatest common divisor, and thus obtain (\ref{eq:min_laeuf}).
      For (\ref{eq:min_laeuf_bar}) we do the same with $(\bar{\mathrm{I}})$ and $(\bar{\mathrm{II}})$.
      An arbitrary glider can be written as the glider defined by either (\ref{eq:min_laeuf}) or (\ref{eq:min_laeuf_bar}) times a Laurent polynomial in $\varia$.     
    \end{enumerate}
  \end{proof}
\end{prop}

\begin{remark}\rm
  \label{rem:glider}
  We could extend our definition of gliders to observables with polynomial eigenvalues $\lambda$. These observables would be mapped to products of translates of themselves. We can show that this extension would not yield any new gliders: A CSCA $\sca$ has to fulfill $\det\sca=1$. With $(I)$ to $(\bar{II})$ we get
  \begin{eqnarray*}
    \det \sca \cdot\left(\xi\wedge\bar\xi\right)^2 & = & \left(\sca_{11}\sca_{22}+\sca_{12}\sca_{21}\right)\cdot\left(\xi\wedge\bar\xi\right)^2\\
    & = & \left(\lambda\xi_+\bar\xi_-+\bar\lambda\bar\xi_+\xi_-\right) \cdot\left(\bar\lambda\xi_+\bar\xi_-+\lambda\bar\xi_+\xi_-\right)\\
    && - \left(\left(\bar\lambda+\lambda\right)\xi_+\bar\xi_+\right) \cdot \left(\left(\lambda+\bar\lambda\right)\xi_-\bar\xi_-\right)\\
    & = & \lambda\bar\lambda\left(\xi_+^2\bar\xi_-^2+2\xi_+\xi_-\bar\xi_+\bar\xi_-+\bar\xi_+^2\xi_-^2\right)\\
    & = & \lambda\bar\lambda\left(\xi\wedge\bar\xi\right)^2\\
    \Leftrightarrow\quad \det \sca & = & \lambda\bar\lambda \sollgleich 1\\
    \Leftrightarrow\qquad\:\: \lambda & = &\varia^n.
  \end{eqnarray*}
  The only possible solutions are $\lambda=\varia^{\pm n}$, $n \in \NN$, because by Proposition \ref{mono} in $\fring$ only monomials have inverse elements. We know that $\bar\xi$ is always a glider to $-n$ so we only look at positive $n$. $\varia^0$ is excluded, because there is no propagation then.
\end{remark}

\begin{prop}\label{gliderCQCA}
   A minimal $\xi \in \fring^2$ is a glider for a CSCA $\sca$ with eigenvalue $\varia^{n}$, $n\in \mathbb{Z}\backslash 0$ if and only if is a divisor of $\varia^{-n}+\varia^{n}$. 
   \begin{proof}
    First let us assume, that $\xi$ is minimal, and $\xi\wedge\bar\xi$ is divisor of $\varia^{-n}+\varia^{n}$. For $\xi$ to be a glider with eigenvalue $\varia^n$, $\sca \xi=\lambda \xi$ has to hold. Therefore the division in Equations (\ref{eq:a11}) to (\ref{eq:a22}) has to be valid. For the Equations (\ref{eq:a12}) and (\ref{eq:a21}) this is obviously true. For the other two equations we use a simple trick:
    \begin{eqnarray*}
      \sca_{11}\left(\xi\wedge\bar\xi\right) & = & \varia^n\xi_+\bar\xi_-+\varia^{-n}\bar\xi_+\xi_-\\
      &=&\varia^n\xi_+\bar\xi_-+\varia^{-n}\bar\xi_+\xi_-+\underbrace{(\varia^{-n}\xi_+\bar\xi_-+\varia^{-n}\xi_+\bar\xi_-)}_{=0}\\
      &=&(\varia^{-n}+\varia^n)\xi_+\bar\xi_-+\varia^{-n}\left(\xi\wedge\bar\xi\right)
    \end{eqnarray*}
    Now it is apparent, that $\xi\wedge\bar\xi$ also divides the right hand side of (\ref{eq:a11}) if it is a divisor of $\varia^{-n}+\varia^n$. For (\ref{eq:a22}) an analogous argument holds.
    
    Now let us show the converse. The Laurent polynomials form an Euclidean ring, which implies that the extended Euclidean algorithm is applicable \cite{simonreed}. Since $\xi$ is minimal, the greatest common divisor of $\xi_+$ and $\xi_-$ is  1, according to the extended Euclidean algorithm we can chose $\eta_+$ and $\eta_{-}$ such that
    \begin{equation*}
      \xi_+ \eta_{-} + \xi_{-}\eta_{}+=\gcd(\xi_+,\xi_{-})=1.
    \end{equation*}
    Then we have:
    \begin{eqnarray*}
      \left(\varia^{-n}+\varia^n\right) & = & \left(\varia^{-n}+\varia^n\right)\cdot\underbrace{\left(\xi_+\eta_-+\xi_-\eta_+\right)}_{\mathrm{=1,\,as\,\xi\,min.}}\cdot \underbrace{\left(\bar\xi_+\bar\eta_-+\bar\xi_-\bar\eta_+\right)}_{\mathrm{=1,\,as\,\bar\xi\,min.}}\\
      & = & \left(\varia^{-n}+\varia^n\right)\cdot\left(\xi_+\bar\xi_+\eta_-\bar\eta_-+\xi_-\bar\xi_+\eta_+\bar\eta_-+ \xi_+\bar\xi_-\eta_-\bar\eta_++\xi_-\bar\xi_-\eta_+\bar\eta_+\right)\\
      & = &\left(\varia^{-n}+\varia^n\right)\xi_+\bar\xi_+\eta_-\bar\eta_- + \left(\varia^{-n}+\varia^n\right)\cdot\xi_-\bar\xi_-\eta_+\bar\eta_+\\
      & & + \varia^{-n}\xi_-\bar\xi_+\eta_+\bar\eta_- + \varia^n\xi_-\bar\xi_+\eta_+\bar\eta_- \underbrace{+\varia^{-n}\xi_+\bar\xi_-\eta_+\bar\eta_- + \varia^{-n}\xi_+\bar\xi_-\eta_+\bar\eta_-}_{=0} \\
      & & + \varia^{-n}\xi_+\bar\xi_-\eta_-\bar\eta_+ + \varia^n\xi_+\bar\xi_-\eta_-\bar\eta_+ \underbrace{+\varia^{-n}\bar\xi_+\xi_-\eta_-\bar\eta_+ + \varia^{-n}\bar\xi_+\xi_-\eta_-\bar\eta_+}_{=0}\\
      & = & \underbrace{\left(\varia^{-n}+\varia^n\right)\xi_+\bar\xi_+}_{=\sca_{12}\left(\xi\wedge\bar\xi\right)} \eta_-\bar\eta_- + \underbrace{\left(\varia^{-n}+\varia^n\right)\xi_-\bar\xi_-}_{=\sca_{21}\left(\xi\wedge\bar\xi\right)} \eta_+\bar\eta_+\\ 
      & & +\underbrace{\left(\varia^{-n}\xi_+\bar\xi_-+\varia^n\xi_-\bar\xi_+\right)}_{=\sca_{22}\left(\xi\wedge\bar\xi\right)} \eta_+\bar\eta_- + \underbrace{\left(\varia^n\xi_+\bar\xi_-+\varia^{-n}\xi_-\bar\xi_+\right)}_{=\sca_{11}\left(\xi\wedge\bar\xi\right)} \eta_-\bar\eta_+ \\
      & & + \varia^{-n}\eta_+\bar\eta_-\left(\xi\wedge\bar\xi\right) +\varia^{-n}\eta_-\bar\eta_+\left(\xi\wedge\bar\xi\right)
    \end{eqnarray*}
    which implies that $\xi \wedge \bar\xi$ divides $(u^{-n} +u^{n})$.
  \end{proof}
\end{prop}

We have shown in Proposition \ref{prop:glider_sca_props} that $\tr\sca =\varia^{-n}+\varia^{n}$ is a necessary condition for $\sca$ to have gliders. The following proposition shows that this condition is also sufficient.
\begin{prop}
  A CSCA possesses gliders with eigenvalues $\lambda_+=\varia^n$ and $\lambda_-=\bar\lambda_+=\varia^{-n}$ if and only if $\tr \sca=\varia^{-n}+\varia^{n}$. 
  \begin{proof}
    The ``only if'' part was already shown in Proposition \ref{prop:glider_sca_props}.
    
    We now assume that $\tr \sca=\varia^{-n}+\varia^{n}$ and use this to evaluate the characteristic polynomial of $\sca$. Using $\det\sca=1$ we get
    \begin{equation*}
      \lambda^2 + \lambda\cdot(\varia^{-n}+\varia^{n})+ 1 = 0,
    \end{equation*}
    which is solved by $\lambda_{\pm}=\varia^{\pm n},\,n\in\NN$. Thus the CSCA possesses gliders.
  \end{proof}
\end{prop}

Now that we know the conditions for the existence of gliders, we want to know how and when they can be connected. Consider an arbitrary CSCA $\sca$ with gliders and a second CSCA $\scb$. If we transform $\sca$ by conjugating with $\scb$ we get $\widetilde\sca=\scb^{-1} \sca \scb$ which has by
\begin{equation*}
  \tr\widetilde\sca=\tr(\scb^{-1} \sca \scb)=\tr(\Id\sca)=\tr\sca
\end{equation*}
the same trace as $\sca$ and thus is a glider automorphism, too.
What is maybe more surprising is that the converse is also true for gliders of propagation speed one: any CSCA with one-step gliders is equivalent to the standard-glider CSCA $\simpleglidersca$ (\ref{eq:simpleglider}) by the equivalence relation  $\widetilde\sca=\scb^{-1}\simpleglidersca \scb$.
\begin{thm}\label{thm:glider_equiv}
  Let $\xi=\left(\xi_+,\,\xi_-\right)\,\in\,\fring^2$ be minimal. Then the following three statements are equivalent:
  \begin{enumerate}
    \item There is a CSCA $\sca$ with $\sca\xi=\varia\xi$.
    \item There is a CSCA $\scb$ with $\scb\xi=\tbinom{1}{\varia}$.
    \item $\xi\wedge\bar\xi=\left(\varia^{-1}+\varia\right)$.
  \end{enumerate}
  \begin{proof}
    $3 \Leftrightarrow 1$ has already been shown in Proposition \ref{gliderCQCA}.
    
    $1\Rightarrow 2$: We assume that 1 (and therefore also 3) is true and analyze the conditions this imposes on $\scb$: We know that $\scb=\bar \scb$. We start with the assumption $\scb\xi=\tbinom{1}{\varia}$ and obtain
    \begin{eqnarray}
      \scb_{11}\left(\xi\wedge\bar\xi\right) & = & \xi_-+\bar\xi_-, \label{eq:b11}\\
      \scb_{12}\left(\xi\wedge\bar\xi\right) & = & \xi_++\bar\xi_+,\\
      \scb_{21}\left(\xi\wedge\bar\xi\right) & = & \left(\varia^{-1}\xi_-+\varia\bar\xi_-\right),\\
      \scb_{22}\left(\xi\wedge\bar\xi\right) & = & \left(\varia^{-1}\xi_++\varia\bar\xi_+\right). \label{eq:b22}
    \end{eqnarray}
    First we need to show that the matrix $\scb$ actually exists, i.e., that all the right sides of the equations (\ref{eq:b11}) to (\ref{eq:b22}) can be divided by $\xi \wedge \bar \xi$. This follows from the same argument as in Proposition \ref{gliderCQCA} if $\tr\sca=\varia^{-1}+\varia$, which is given by $1\Leftrightarrow 3$. For the matrix to be a CSCA the determinant has to be one. This is also true and can be shown by direct computation: 
    \begin{eqnarray*}
      \det \scb & = & \frac{1}{\left(\xi\wedge\bar\xi\right)^2}\cdot\left(\left(\xi_-+\bar\xi_-\right)\cdot\left(\varia^{-1}\xi_++\varia\bar\xi_+\right)
      +\left(\varia^{-1}\xi_-+\varia\bar\xi_-\right)\cdot \left(\xi_++\bar\xi_+\right)\right)\\
      &=&\frac{1}{\left(\xi\wedge\bar\xi\right)^2}\cdot\left(\varia^{-1}+\varia\right)\cdot\left(\xi\wedge\bar\xi\right)\\
      &=&\frac{\left(\varia^{-1}+\varia\right)}{\xi\wedge\bar\xi}= 1
    \end{eqnarray*}
    This step only works for $\xi\wedge\bar\xi=(\varia^{-1}+\varia)$ which corresponds to one step gliders.
    Later on we will consider gliders with higher propagation speed and give counterexamples to similar notions of equivalence for their automata.
       
    $2 \Rightarrow 3$:
    \begin{displaymath}
      \xi\wedge\bar\xi=\underbrace{\det \scb}_{=1}\cdot(\xi\wedge\bar\xi)=\scb\xi\wedge \scb\bar\xi=\tbinom{1}{\varia}\wedge\tbinom{1}{\varia^{-1}} = \varia^{-1}+\varia
    \end{displaymath}
    This completes our proof.
  \end{proof}
\end{thm}
  
We will now show, that for automata with higher propagation speed the gliders are not equivalent in the above sense. It is apparent from proposition \ref{gliderCQCA}, that for a fixed $n>1$ there are gliders with different wedge products. These can not be transformed into each other, because the wedge product is invariant under transformation with CSCAs (see Theorem \ref{thm:glider_equiv}, part 2). Another way to see that there are different types of $n$-step glider automata is the fact, that we always have automata which are powers of one-step automata and also automata, whose roots are not CSCAs. These can not be transformed into each other. But even automata for gliders with the same wedge product can not always be connected by a third CSCA. To show this, it is sufficient to find two phase space vectors $\xi$ and $\eta$ with $\xi\wedge\bar\xi=\eta\wedge\bar\eta$ and $\xi\wedge\bar\xi$ dividing $(\varia^{-n}+\varia^{n})$ for some $n$ which can not be transformed into each other by a CSCA. We choose $\xi=\tbinom{1}{\varia+\varia^2}$ and $\eta=\tbinom{1+\varia}{\varia^2}$. Their wedge product $\xi\wedge\bar\xi=\eta\wedge\bar\eta=\varia^{-2}+\varia^{-1}+\varia+\varia^2$ divides $(\varia^{-3}+\varia^{3})$. It is  a valid wedge product for a $3$-step glider. If an automaton $\scb$ with $\scb\eta=\xi$ existed, it would have to fulfill the equations
  \begin{eqnarray*}
    \scb_{11}\eta_++\scb_{12}\eta_-&=&\xi_+=1\\
    \scb_{11}\bar\eta_++\scb_{12}\bar\eta_-&=&\bar\xi_+=1.
  \end{eqnarray*}
  From these we get 
  \begin{eqnarray*}
    \scb_{11}(\xi\wedge\bar\xi)&=&\bar\eta_-+\eta_-\\
    \Leftrightarrow \scb_{11}(\varia^{-2}+\varia^{-1}+\varia+\varia^2)&=&\varia^{-2}+\varia^2
  \end{eqnarray*}
  which can not be solved by any $\scb_{11}\in\subfring$.
\subsubsection{Periodic automata}
\label{sec:periodic}
CSCAs whose matrices have a trace independent of $\varia$ show periodic behavior. 

\begin{prop}
  \label{prop:periodic_sca}
  A CSCA $\sca$ is periodic with period $c+2$ if $\tr\sca=c$ for $c\in\ZZ_2$.
  \begin{proof}
    By the Cayley-Hamilton theorem we get $\sca^2=\sca\cdot\tr\sca+\Id=c\sca + \Id$ and thus $\sca^2=\Id$ for $c=0$ and $\sca^3=\sca^2+\sca=2\sca+\Id=\Id$ for $c=1$.
  \end{proof}
\end{prop}

\begin{prop}
  \label{prop:ev_one_periodic}
  Let $\sca$ be a CSCA and $\xi\in\fring^2$ a non-zero phase space vector. If $\sca\xi=\xi$ holds, then $\sca$ is periodic with period two. 
  \begin{proof}
    We use $\sca^2=\sca\cdot\tr\sca+\Id$ and $\sca\xi=\xi$:
    \begin{equation*}
      \begin{array}{rrcl}
        &\sca^2\xi&=&\xi\\
        \Leftrightarrow&(\sca\cdot\tr\sca)\xi+\Id\xi&=&\xi\\
        \Leftrightarrow&\xi(\tr\sca)&=&0\\
        \Leftrightarrow&\tr\sca&=&0
      \end{array}
    \end{equation*}
    Thus $\sca$ is of period two by Proposition \ref{prop:periodic_sca}.
  \end{proof}
\end{prop}
%%%%
\subsubsection{Fractal automata}
\label{sec:frac}
All CSCAs that are neither periodic nor have gliders show a fractal time evolution. Fractal means, that the graph of the spacetime evolution of one cell observables is self similar in the limit of infinitely many timesteps. We will cover this type of CSCAs in detail in a future publication \cite{Nesme}. Here we only give an example to illustrate the self similarity. The evolution of the automaton
\begin{equation}
  \label{eq:fractal}
  \fractalsca=\left(\begin{array}{cc}
      \varia^{-1}+1+\varia&1\\
      1&0
    \end{array}
  \right)
\end{equation}
is shown in Figure \ref{fig:frac}.
\begin{figure}[htbp]
  \begin{center}
    \includegraphics[width=\textwidth, angle=180]{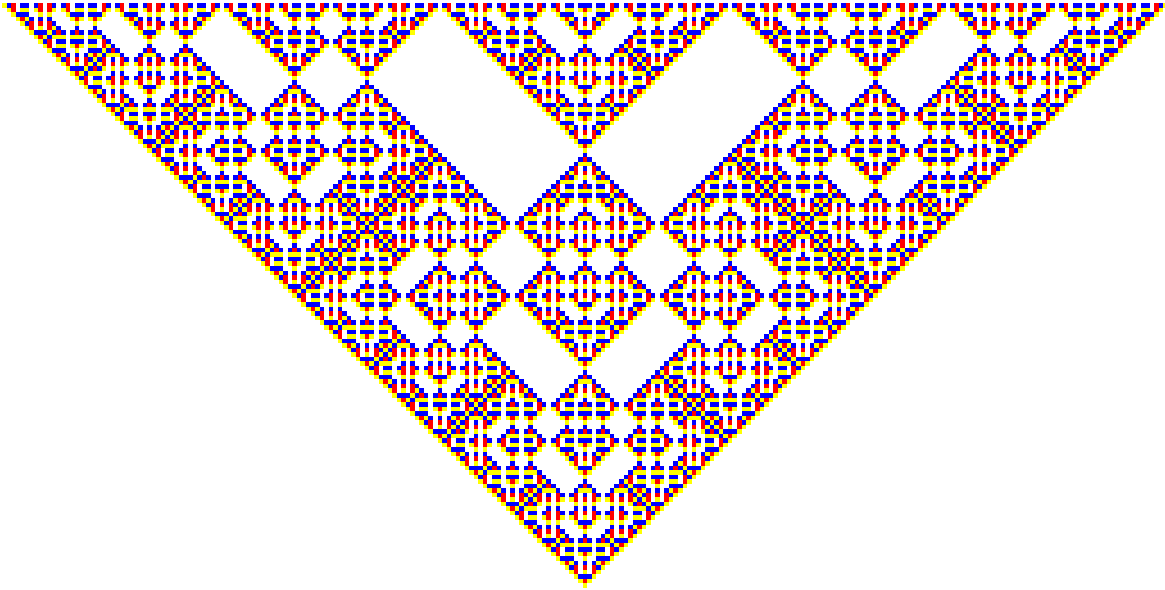}
    \caption{Time evolution of a fractal CSCA (\ref{eq:fractal}). Time increases upwards. The different colors illustrate the different Pauli matrices.}
    \label{fig:frac}
  \end{center}
\end{figure}

Nevertheless, we will state one short lemma, that we will need later on to prove the convergence of product states.

\begin{lem}\label{lem:fractal_power}
  If a CSCA $\sca$ is fractal, so is $\sca^n$ for $n \in \NN$.
  \begin{proof}
    We prove that $\sca^n$ is fractal by showing that it can neither be periodic nor have gliders. Obviously, no power of $\sca$ can be periodic if $\sca$ is not periodic, so only the glider case remains. If $\sca^n$ has minimal gliders $\xi$, $\bar\xi$ then by
    \begin{equation*}
       \sca^n (\sca \xi)=\sca \sca^n \xi=\lambda (\sca \xi)
    \end{equation*}
    $\sca\xi$ is also a glider for $\sca^n$ with the same eigenvalue $\lambda$ and thus a multiple of $\xi$.
    Hence $\sca\xi=\tilde\lambda\xi$ holds for a monomial\footnote{By Remark \ref{rem:glider} the only possible eigenvalues are monomials. } $\tilde\lambda$. If $\tilde\lambda=1$ the automaton would be periodic by Proposition \ref{prop:ev_one_periodic} which is already ruled out. For $\tilde\lambda=\varia^{\pm n}$, which is the only other possibility, $\sca$ has gliders. Thus for $\sca^n$ to have gliders $\sca$ has to have gliders. This is a contradiction to the assumption that $\sca$ is fractal. So any power of a fractal CSCA is always fractal.
  \end{proof}
\end{lem}
%%%
\section{Invariant states and convergence}\label{sec:stationary} 
In this section we will consider different types of states on the spin chain and their evolution under CQCA action. Our focus lies on the search for invariant states and the convergence of other states towards invariant states. We consider special sets of states: product states, stabilizer states and quasifree states. Because CQCAs act in a translation-invariant manner, it is natural to look at translation-invariant states. We will therefore only consider those.

The only state that we know to be invariant for all CQCAs (and all QCAs) is the tracial state, which vanishes on all finite Pauli products except the identity and is defined as the limit of states that have the density operator $\frac{1}{2} \Id$ for each tensor factor. We strongly suspect this state to be the only invariant state for fractal CQCAs, but have no complete proof yet.
\subsection{Invariant states for periodic CQCAs}
  \label{sec:inv_periodic}
  It is obvious, that all states are periodic under the action of periodic CQCAs. Therefore a state is either invariant or does not converge at all. Finding invariant states for periodic automata is in general very easy. For example, if a CQCA satisfies $T^p=\mathds{1}$ for some finite $p\in\mathds{N}^{+}$, then all states of the type $\frac{1}{p}\sum_{n=1}^{p} \omega\circ T^n$ are invariant. Moreover, any invariant state is of this form, because for $\omega$ invariant $\omega=\frac{1}{p}\sum_{n=1}^{p} \omega\circ T^n$. Finding pure invariant states is more complicated. In Section \ref{sec:stab_inv} we show, that for some period-two CQCAs pure invariant stabilizer states exist. For CQCAs without propagation there also exist invariant product states. Namely, if one Pauli matrix is left invariant by such a CQCA, then the state that gives expectation value one on this matrix and vanishes on the others is left invariant by this CQCA. A state with the same expectation value for all Pauli matrices is always left invariant by a non-propagating CQCA.
  %%%%
\subsection{Invariance and convergence of product states for non-periodic automata}
\label{sec:stat_conv_prod}
In this section we will only consider states which are translation-invariant product states with respect to single cell systems. They are of the form $\omega(\weyl{\xi})=\prod_i\Phi(\weyl{\xi_i})$, where $\Phi$ is a state on a single cell. 

\begin{prop}
  \label{prop:conv_product_glider}
  For a glider CQCA $\CQCA$ there exist no translation-invariant product states that are T-invariant except the tracial state. All other translation-invariant product states weakly converge to $\CQCA$-invariant states 
  \begin{equation}
    \omega_\infty=\lim_{n\to\infty}\omega\circ\CQCA^n.
  \end{equation}
  with the following properties:
  \begin{itemize}
    \item $\omega_\infty(\weyl{\xi})=1$, if $\weyl{\xi}=\Id$,
    \item $|\omega_\infty(\weyl{\xi})|<1$, if $\weyl{\xi}$ is a product of gliders,
    \item $\omega_\infty(\weyl{\xi})=0$ otherwise.
  \end{itemize}
  If $\Phi(\sigma_j)=1$ for some $j$, then $\omega_\infty$ is the tracial state. 
  \begin{proof}
    First we prove the non-existence of invariant product states. A state is invariant if $\omega\circ\CQCA=\omega$, i.e., $\omega(\CQCA A)=\omega(A),\,\forall A\in\FA$ has to be fulfilled. We require the automaton to be non-periodic which implies a finite (non-zero) propagation, so at least two\footnote{The image of the third Pauli matrix is always determined by the product of the other two.} of the Pauli matrices are mapped to tensor products of at least three Pauli matrices\footnote{The image has to be a tensor product of at least tree Pauli matrices, because an identity in the middle is not allowed.}. Thus for these two Pauli matrices $\weyl{\xi}$
    \begin{eqnarray*}
      \Phi(\weyl{\xi})&=&\omega(\weyl{\xi})\\
      &=&\omega(\CQCA\weyl{\xi})\\
      &=&\prod_{i\in\CN}\Phi(\weyl{(\sca\xi)_i})
    \end{eqnarray*}
    holds. For all Pauli matrices $\sigma_i$, $|\Phi(\sigma_i)|\le1$. If $|\Phi(\sigma_i)|=1$, then $\Phi(\sigma_j)=0,\,\forall j\ne i$. Now lets assume that $\Phi(\sigma_i)\ne0$ for some $i$. The image of a single cell observable has to include at least two different types of Pauli matrices (different from the identity, e.g. $\sigma_1$ and $\sigma_2$). Else $\xi_+=0$ or $\xi_-=0$ or $\xi_+=\xi_-$, each case implying common divisors.
    
    Let us consider the case when there exists a Pauli matrix which is not mapped to a tensor product. It can not be mapped to itself, because then by Proposition \ref{prop:ev_one_periodic} the automaton would be periodic. So it has to be mapped to another Pauli matrix which has to expand in the next step. Thus we only need to consider the case of expanding Pauli matrices. The image has to consist of more than one kind of Pauli matrices, so $|\Phi(\sigma_i)|=1$ is ruled out as an invariant state. Moreover, the image can not contain the original Pauli matrix even once, because $\Phi(\sigma_{i_k})=(\Pi_{j\in\CN\setminus k}\Phi(\sigma_{i_j}))\Phi(\sigma_{i_k})$ implies $|\Phi(\sigma_{i_j})|=1,\,\forall i_j$ which is already ruled out. So no Pauli matrix may occur in its own image, particularly not in the central position. If we only consider those central positions, we get a local automaton. It is an easy calculation to show, that all of these automata, which map no Pauli matrix onto itself have trace one. So the trace of our CQCA contains a constant which is a contradiction to the condition that it has gliders. So $\Phi(\sigma_i)=0,\,i=1,2,3$ and the only invariant state is the one with the density matrix $\rho_\Phi=\frac{1}{2}\Id$, i.e., the tracial state. 
   
    In the following we consider the convergence properties of product states. Clearly, it suffices to establish the convergence of Pauli products. First, suppose that some $|\Phi(\sigma_j)|=1$. Then, according to Lemma \ref{lem:finite_pauli}, for all Pauli products different from $\Id$, and all times except at most one, the evolved product will contain a Pauli matrix different from $\sigma_j$, and hence will have zero expectation in $\omega\circ \CQCA^n$. Hence $\omega_\infty$ is the tracial state. 

    To treat the remaining cases, we hence assume from now on that $|\Phi(\sigma_j)|\leq\lambda < 1$ for all three $j$. Hence if some Pauli product has $k$ factors different from $\Id$, its expectation in $\omega$ is at most $\lambda^k$. So let $k(t)$ be the number of non-identity factors in the $t^{\rm th}$ iterate of some Pauli product. If $k(t)$ diverges, we have nothing to prove. So we may assume from now on that there is a constant $k_{max}<\infty$ such that $k(t)<k_{max}$ infinitely often. We focus on the subsequence for which this is the case.  
    
    If the overall length of the Pauli product (largest-smallest degree) remains finite, and since we have assumed the absence of periodic finite configurations, then we must have a glider for some power of $\CQCA$. As argued in the proof of Lemma \ref{lem:fractal_power}, this is also a glider for $\CQCA$. In fact, for any product of glider elements, the left going and the right going gliders will eventually be separated, and from that point onwards the $\omega$-expectation will not change anymore. Hence the limit does exist and will be some number of modulus $<1$.

    Hence we need only consider the case that $k_{max}$ is finite, but the positions of non-identity Pauli factors get more and more spread out. This is only possible, if some Pauli products near the edges of the given product have a similar property: no sub-product which gets widely separated from the rest is allowed to develop configurations with unbounded $k_{max}$, because then the overall bound could not hold. Thus we again find bounded configurations with smaller $k_{max}$, which may again be smaller gliders, or split up even further. By downwards induction we just get a complete decomposition into gliders for any configuration with non-divergent $k_{max}$.  This completes our proof.
  \end{proof}
\end{prop}

\begin{prop}
  \label{prop:conv_product_fractal}
  Under the action of fractal CQCAs all product states converge to the tracial state. 
  \begin{proof}
    To show convergence for the fractal case, we need the results from Section \ref{sec:frac} and Appendix~\ref{app:frac}. They state, that a tensor product with only one kind of Pauli matrices can occur only once in the history of a fractal CQCA and that the number of non-identity tensor factors grows unbounded. With the arguments used in the proof of Theorem III.1 this means that for fractal CQCAs any given product state converges to the state which gives zero expectation value for any non-trivial tensor products of Pauli operators, i.e., to the tracial state.
  \end{proof}
\end{prop}
%%%%%
\subsection{Invariance and convergence of stabilizer states}
\label{sec:stab_inv_conv}
In this chapter we consider pure translation-invariant stabilizer states. A stabilizer state is the common eigenstate of an abelian group of operators (usually tensor products of Pauli matrices) called the stabilizer group $\stab$. For those not familiar with the stabilizer formalism we give a short example.
\begin{example}\rm
  \label{exam:bell}
  The stabilizer group $\stab=\langle \sigma_1\otimes \sigma_1,\sigma_3\otimes \sigma_3\rangle$ stabilizes the Bell state $\psi=1/\sqrt{2}(\ket{1,1}+\ket{0,0})$. We check this by applying the stabilizer generators to the state (normalization is omitted):
  \begin{eqnarray*}
    (\sigma_1\otimes \sigma_1)\psi&=&(\sigma_1\otimes \sigma_1)(\ket{1}\otimes\ket{1})+(\sigma_1\otimes \sigma_1)(\ket{0}\otimes\ket{0})\\
    &=&\ket{0,0}+\ket{1,1}=\psi\\
    (\sigma_3\otimes \sigma_3)\psi&=&(\sigma_3\otimes \sigma_3)(\ket{1}\otimes\ket{1})+(\sigma_3\otimes \sigma_3)(\ket{0}\otimes\ket{0})\\
    &=&(-\ket{1})\otimes(-\ket{1})+\ket{0}\otimes\ket{0}\\
    &=&(-1)^2\ket{1,1}+\ket{0,0}=\psi
  \end{eqnarray*}
  \fin
\end{example}
Now we move on to translation-invariant stabilizer states. Consider the stabilizer group $\stab=\langle \weyl{\hat\tau^x\xi}, x \in \ZZ \rangle$ generated by the translates of a Pauli tensor product $\weyl{\xi}$, where $\hat\tau^x$ denotes the $x$-th power of the phase space translation. The unique state $\omega $ satisfying $\omega(A)=1$ for all $A \in \stab$ is called the pure translation-invariant stabilizer state corresponding to the stabilizer $\stab$ \cite{SchlingemannCQCA}. In the following we will refer to the stabilizer group simply as ``stabilizer''.
\subsubsection{Invariance}
\label{sec:stab_inv}
Let us first search which CQCAs leave pure translation-invariant stabilizer states invariant. 
\begin{prop}
  Only periodic CQCAs can leave pure translation-invariant stabilizer states invariant. For each such stabilizer state with stabilizer $\stab=\langle\weyl{\hat\tau^x\xi},\,x\in\ZZ\rangle$ the CQCAs that leave the state invariant are periodic and form a group. The corresponding CSCAs are
  \begin{equation}
    \sca_\xi(a)=\left(\begin{array}{cc}
                   1+a\xi_+\xi_-&a\xi_+^2\\
                   a\xi_-^2&1+a\xi_+\xi_-
                 \end{array}
           \right),
  \end{equation}
  for a centered palindrome $a$.
  \begin{proof}
    The invariance condition for a stabilizer state with stabilizer $\stab=\langle\weyl{\hat\tau^x\xi},\,x\in\ZZ\rangle$ is 
    \begin{equation*}
      \sca\xi=\alpha\xi,\,\alpha\in\fring.
    \end{equation*}
    In \cite{SchlingemannCQCA} it was proved that for every translation-invariant stabilizer state $\omega$ with stabilizer $\stab=\langle\weyl{\hat\tau^x\xi},\,x\in\ZZ\rangle$ there exists a CQCA $B$, which maps $\weyl{0,1}$ to $\weyl{\xi}$. Thus the condition becomes
    \begin{equation*}
      \sca{\bf b}\tbinom{0}{1}=\alpha{\bf b}\tbinom{0}{1}.
    \end{equation*}
    If we know automata that leave the ``all spins up'' state invariant, we can construct automata for arbitrary translation-invariant pure stabilizer states via
    \begin{equation*}
      \begin{array}{crcl}
        &\phantom{\dbinom{0}{0}}\sca\tbinom{0}{1}&=&\alpha\tbinom{0}{1}\\
        \Leftrightarrow&\underbrace{{\bf b}\sca{\bf b}^{-1}}_{\sca_\xi}\underbrace{{\bf b}\tbinom{0}{1}}_{\xi}& = & \alpha{\bf b}\tbinom{0}{1}\\
      \Leftrightarrow&\sca_\xi\xi&=&\alpha\xi.
      \end{array}
    \end{equation*}
    The only type of CQCAs that leave the ``all spins up'' state invariant are the shear transformations, which are represented by matrices 
    \begin{equation*}
      \sca=\shear{a}
    \end{equation*}
    with some palindrome $a$.
    A direct computation shows that the factor $\alpha$ has to be invertible. As we are in characteristic two and deal with centered automata, only $\alpha=1$ is possible. For a general translation-invariant stabilizer state the CQCAs that leave it invariant are represented by the CSCAs with matrices
    \begin{equation*}
      \sca_\xi(a)=\left(\begin{array}{cc}
                   1+a\xi_+\xi_-&a\xi_+^2\\
                   a\xi_-^2&1+a\xi_+\xi_-
                 \end{array}
           \right).
    \end{equation*}
    The product of two such CQCAs also leaves the state invariant with $\sca_\xi(a)\sca_\xi(b)=\sca_\xi(a+b)$. The inverse of each periodic CQCA is one of its powers, so it is also included in this set. We therefore have a group of period two CQCAs for each translation-invariant pure stabilizer state that leave this state invariant. 
  \end{proof}
\end{prop}
\subsubsection{Convergence}
We now consider an arbitrary translation-invariant stabilizer state $\omega_{\weyl{\xi}}$ and the respective stabilizer $\stab=\langle\weyl{\hat\tau^x\xi},\,x\in\ZZ\rangle$.
There exists a CQCA $B$ satisfying $B(\weyl{\xi})=\sigma_3$. We rewrite our state in terms of the $\sigma_3$-state
\begin{equation}
 \omega_{\weyl{\xi}}=\omega_{\sigma_3}\circ B^{-1}
\end{equation}
and use
\begin{eqnarray*}
&&\omega_{\weyl{\xi}}(T^n (\bigotimes \sigma_i))\\
&=&\omega_{\sigma_3}(B^{-1}\CQCA^n (\bigotimes \sigma_i))\\
&=&\omega_{\sigma_3}((B^{-1}\CQCA B)^{n} B^{-1} (\bigotimes \sigma_i)).
\end{eqnarray*}
$B^{-1}\CQCA B$ is also a CQCA of the same type (same trace) as $\CQCA$ and $B^{-1} (\bigotimes \sigma_i)$ is a tensor product of Pauli matrices. $\omega_{\sigma_3}$ is a product state and thus converges according to Propositions \ref{prop:conv_product_glider} and \ref{prop:conv_product_fractal} for glider and fractal CQCAs.
%%%%%%%
\subsection{Stationary quasifree states and convergence of quasifree states for one-step glider CQCAs}
\label{sec:stat_conv_qf}
In the previous sections we discussed 
stationary states and convergence of states
under general CQCA actions.
In this section we consider the particular
glider CQCA $\glider$ which is represented by the CSCA
\begin{equation}
  \label{eq:glider}
  \glidersca=\left(\begin{array}{cc}
      1&\varia^{-1}+\varia\\
      1&\varia^{-1}+1+\varia
    \end{array}
  \right),
\end{equation}
because for this CQCA we can obtain new types of 
stationary states and new convergence results by 
employing the Araki-Jordan-Wigner transformation.
Using the Theorem for glider equivalence \ref{thm:glider_equiv} we can construct invariant states for all automata with gliders that move one step in space every timestep. According to this theorem for any speed-one glider CQCA $B$, there is a CQCA $A$ such that $B=AGA^{-1}$, and  if $\omega$ is a $\glider$-invariant state, then $\omega \circ A^{-1}$
will be $B$-invariant.
%%%%%%%%%
\subsubsection{Araki-Jordan-Wigner transformation}\label{sec:AJW}
\label{sec:ajw}
The Jordan-Wigner transformation is a way to map 
a finite spin-chain algebra $\mathcal{M}_2^{\otimes N}$ to the 
algebra of a finite fermion chain. 
This method has been extensively used
in solid state physics \cite{LSM,BMD,BM}. 
However, the method cannot be carried over directly
to two-sided infinite chains.
One has to introduce an additional infinite ``tail-element'' 
for the transformation to work. 
This extended transformation was introduced by Araki
in his study of the two-sided infinite XY-chain \cite{araki-ajw},
and it is sometimes referred to as the Araki-Jordan-Wigner
construction.

The C$^*$-algebra describing a two-sided infinite fermion chain is the algebra
$\mathcal{F}$=CAR($\ell^{2}(\mathbb{Z})$), i.e., it is 
the C$^{*}$-algebra generated by $\mathbbm{1}$ and and the annihilation and creation operators
$\{c^{\phantom*}_x \}_{x \in \mathbb{Z}}$ and $\{c^{*}_{x}\}_{x \in \mathbb{Z}}$, 
satisfying the canonical anticommutation relations:
\begin{eqnarray*}
      c^{\phantom*}_{x}c^*_{y}+c^*_{y}c^{\phantom*}_{x}=
        \delta_{x,y} \mathbbm{1}, \quad \quad
      c^{\phantom*}_{x}c^{\phantom*}_{y}+c^{\phantom*}_{y}c^{\phantom*}_{x}=0.
\end{eqnarray*}
The translation automorphism $\trans{\mathcal{F}}$ on this algebra is defined by
$\trans{\mathcal{F}}(c_{x})= c_{x+1}$. $\mathcal{F}$ is isomorphic to the
observable algebra $\FA$ of the spin chain, but there 
exist no isomorphism $\iota \colon \FA \to \mathcal{F}$ 
that satisfies the property $\iota \circ \trans{} = \trans{\mathcal{F}} \circ \iota $.
This intertwining property would be needed 
to derive the translation invariance of a state $\omega \circ \iota$ on 
$\FA$ from that of $\omega$ on $\mathcal{F}$. 
This problem can be circumvented by the Araki-Jordan-Wigner construction.

The ordinary Jordan-Wigner isomorphism between the
$N$-site spin-chain algebra $M_2^{\otimes N}$
(generated by the finite number of 
Pauli matrices $\{\sigma_1^x,\sigma_2^x,\sigma_3^x \}_{x \in \{0,1, \ldots , N-1\}}$ ), 
and the $N$-site fermion-chain algebra
(generated by $\{c_x, \; c_x^* \}_{x \in \{0,1, \ldots , N-1\}}$), 
is given by
\begin{eqnarray*}
\iota^N_{JW}(\sigma_1^x)&:=&\prod_{y=0}^{x-1} (2c_y^{\phantom*} c_y^{*} -
\mathbbm{1})
(c_x^{*}+c_x^{\phantom*}),\\
\iota^N_{JW}(\sigma_2^x)&:=&\prod_{y=0}^{x-1} (2c_y^{\phantom*} 
c_y^{*} -
\mathbbm{1})
i(c_x^{*}-c_x^{\phantom*}), \\
\iota^N_{JW}(\sigma_3^x)&:=&2c_x^{\phantom*} c_x^{*} - \mathbbm{1}.
\end{eqnarray*} 
However, as we have mentioned, the Jordan-Wigner transformation cannot be
generalized to be a translation-intertwining isomorphism between
the two-sided infinite spin and fermion chains.
In an informal way, one could say that an element of the form
``$\prod_{y=-\infty}^{-1}(2c_y^{\phantom*} c_y^{*} -
\mathbbm{1})$'' would be needed
in the definition of a ``two-sided infinite chain Jordan-Wigner
transformation''. However, $\mathcal{F}$ doesn't  
contain such an element.
The basic idea of the Araki-Jordan-Wigner construction \cite{araki-ajw} 
is to extend the algebra $\mathcal{F}$ with such an 
infinite tail-element. More concretely, one 
defines the C$^*$-algebra $\widetilde{\mathcal{F}}$ to be the extension of
$\mathcal{F}$ by a self-adjoint unitary element $U$ 
satisfying:\footnote{In the literature the symbol $T$ is used
almost exclusively for denoting this unitary element. However, we 
chose to denote it by $U$ to avoid confusion with 
the time-evolution automorphism, which is
denoted by $T$ in this paper.}
\begin{equation} 
    Uc_xU=\begin{cases} 
           \text{ } c_x & \text {if} \; \;  x \ge 0\\
           -c_x &  \text {if} \; \;  x<0
          \end{cases} . \nonumber
  \end{equation}
Clearly, every element of $\widetilde{\mathcal{F}}$ 
can be uniquely written in the form $a+Ub$ with $a,b\in\mathcal{F}$, i.e.
$\widetilde{\mathcal{F}}=\mathcal{F}+U\mathcal{F}$. One can 
extend the translation automorphism 
$\trans{\mathcal{F}}$ to $\widetilde{\mathcal{F}}$ 
through the formula $\tilde{\trans{}}_{\mathcal{F}}(a+Ub):=
\trans{\mathcal{F}}(a)+U(2c_0c_0^*-\mathbbm{1}) \trans{\mathcal{F}}(b)$.
Let us define the elements
\begin{equation*}
  A_x := \begin{cases} \prod_{y=0}^{x-1} 
(2c_{y}^{\phantom*}c_{y}^{*}-\mathbbm{1}) 
&\text{if } x>0 \\
       \mathbbm{1} & \text{if } x=0 \\ 
\prod_{y=x}^{-1} (2c_{y}^{\phantom*}c_{y}^{*}-\mathbbm{1})
       &\text{if } x<0 \end{cases} ,
 \end{equation*}
 and let us introduce the elements
 \begin{equation}
 \widehat{\sigma_1}^x:= UA_x(c_x^{*}+c_x^{\phantom*})   , \quad 
 \widehat{\sigma_2}^x:= UA_x i(c_x^{*}-c_x^{\phantom*}) , \quad
 \widehat{\sigma_3}^x:=(2c_x^{\phantom*} c_x^{*}-\mathbbm{1}). \label{NewPauli}
 \end{equation}
Denote by $\widehat{\mathcal{F}}$ the C$^*$-subalgebra of $\widetilde{\mathcal{F}}$ which is
generated by the elements of
$\{ \widehat{\sigma_1}^x, \widehat{\sigma_2}^x, \widehat{\sigma_3}^x 
\, | \, x \in \mathbb{Z} \}$. 
A direct computation shows that  
the elements defined in (\ref{NewPauli}) having the same 
``spatial index'' $x$ satisfy the Pauli-relations, while any two of 
these elements having two
different spatial indices commute.
Hence $\widehat{\mathcal{F}}$ and $\FA$ are isomorphic,
and an isomorphism is given by the map 
$\Pi: \widehat{\mathcal{F}} \to \FA$, defined as 
\begin{equation*}
\Pi(\widehat{\sigma_1}^x)= \sigma_1^x \, , \; 
\Pi(\widehat{\sigma_2}^x)= \sigma_2^x \, , \;
\Pi(\widehat{\sigma_3}^x)= \sigma_3^x \, , \; \; \; \forall x \in \mathbb{Z}.
\end{equation*}
Moreover, if we denote by $\widehat{\trans{}}_{\mathcal{F}}$ the restriction of
$\widetilde{\trans{}}_{\mathcal{F}}$ to $\widehat{\mathcal{F}}$ then
\begin{equation}
\trans{} \circ \Pi=\Pi \circ \widehat{\trans{}}_{\mathcal{F}}, \nonumber
\end{equation}
i.e., $\Pi$ intertwines the translations of the two algebras.

Let $\omega$ be a translation-invariant state on the
fermion-chain algebra $\mathcal{F}$.
By defining $\widetilde{\omega}(a+Ub):=\omega(a)$ we get a
translation-invariant state on 
$\widetilde{\mathcal{F}}$.\footnote{It is clear that 
by this definition $\widetilde{\omega}$ will be a 
normalized functional on $\widetilde{F}$. The
positivity of $\widetilde{\omega}$ follows from
$\widetilde{\omega}((a+Ub)^*(a+Ub))=\omega(a^*a+b^*b)\ge 0$.}
Restricting this state to $\widehat{\mathcal{F}}$ we get 
a $\widehat{\trans{}}_{\mathcal{F}}$-invariant state 
$\widehat{\omega}$, and 
$\omega^{JW}=\widehat{\omega} \circ \Pi^{-1} $ 
will be a translation-invariant state on the quantum spin-chain.
In this way we can transfer translation-invariant states 
from the fermion-chain to the spin-chain.

Any CQCA automorphism $\CQCA$ can naturally be transferred
to an automorphism on $\widehat{\mathcal{F}}$ commuting with
$\widehat{\trans{}}_{\mathcal{F}}$ by the definition 
$\widehat{\CQCA}:=\Pi \circ \CQCA \circ \Pi^{-1}$.
In the case of the glider CQCA $\glider$ (\ref{eq:glider}) we can do even more. The
transferred automorphism $\widehat{\glider}$ 
(characterized by 
$\widehat{\glider}(\widehat{\sigma_1}^x)= \widehat{\sigma_2}^x $, 
$\widehat{\glider}(\widehat{\sigma_2}^x)=\widehat{\sigma_2}^{x-1}\widehat{\sigma_1}^x\widehat{\sigma_2}^{x+1}$)
can also be extended to an automorphism $\widetilde{\glider} : \widetilde{F} \to \widetilde{F}$ in a translation-invariant way 
($\widetilde{\glider} \circ \tilde{\trans{}}_F= 
\tilde{\trans{}}_F \circ \widetilde{\glider}$) with the following
definition:\footnote{We only have to define the image of $U$ under
$\widetilde{\glider}$, since any element
$\widetilde{f} \in \widetilde{\mathcal{F}}$ 
can be uniquely written as
a linear combination $\widetilde{f}=\widehat{f}_1+U\widehat{f}_2$, where
$\widehat{f}_1, \widehat{f}_2 \in \ \widehat{\mathcal{F}}$.
}
\begin{equation*}
\widetilde{\glider}(U)=iU(c^{\phantom*}_{-1}+c^{*}_{-1})(c^{\phantom*}_0-c^*_0).
\end{equation*}
It is exactly this type of ``translation-invariant extension property'' 
that allows us to find stationary translation-invariant states 
of the glider CQCA by the Araki-Jordan-Wigner method.

Restricting the $\widetilde{\glider}$ 
automorphism to the fermion-chain 
subalgebra $\mathcal{F} \subset \widetilde{\mathcal{F}}$, 
we obtain the automorphism $\glider_{\mathcal{F}}$, 
which acts in the following way:
\begin{equation*}
\glider_{\mathcal{F}}(c_x)=\frac{1}{2}(c^{*}_{x-1}+c^{\phantom*}_{x-1}+c^{*}_{x+1}-
c^{\phantom*}_{x+1}).
\end{equation*}

The automorphism $\glider_{\mathcal{F}}: \mathcal{F} \to \mathcal{F}$
takes an especially simple form in terms of 
majorana operators. These operators are defined as
\begin{equation}
m_{2x}:=i(c^{\phantom*}_x-c_x^*), \quad \quad 
m_{2x+1}:=c^{\phantom*}_x+c_x^*, \label{majorana} 
\end{equation}
for any $x \in \mathbb{Z}$, and they generate $\mathcal{F}$.
The action of $\glider_{\mathcal{F}}$ on these operators is
\begin{equation*}
\glider_{\mathcal{F}}(m_{2x})=m_{2x-2}, \quad \quad \quad 
\glider_{\mathcal{F}}(m_{2x+1})=m_{2x+3}.
\end{equation*}

Clearly, if we find a state $\omega$ on $\mathcal{F}$ that
is both $\trans{\mathcal{F}}\,$- and $\glider_{\mathcal{F}}$-invariant, 
then the Araki-Jordan-Wigner 
transformed state $\omega^{JW}$ will be a $\trans{}$- and $\glider$-invariant 
state on the quantum spin-chain. 
In the next section, we will recall the definition 
of quasifree states on fermion-chains and then determine 
the translation- and $\glider_{\mathcal{F}}$-invariant quasifree states.
In this way, using the Araki-Jordan-Wigner construction,
we can obtain a whole class of translation- and $\glider$-invariant states on 
the spin-chain.  
%%%%%%%%%%%%%%
\subsubsection{Stationary quasifree states}
\label{sec:stat_qf}
A state $\omega: \mathcal{F} \to \mathbb{C}$ is called
{\it quasifree} if it vanishes on odd monomials of majorana operators
\begin{equation*}
\omega(m_{x_{1}} \dots m_{x_{2n-1}})= 0,
\end{equation*}
while on even monomials of majorana operators 
it factorizes in the following form:
\begin{eqnarray*}
\omega(m_{x_{1}} \dots m_{x_{2n}}) = 
\sum\limits_{\pi}{\text{sgn}}(\pi) 
\prod\limits_{l=1}^{n}
\omega(m_{x_{\pi(2l-1)}}m_{x_{\pi(2l)}}), 
\end{eqnarray*}
where the sum runs over all pairings of the set
$\{1,2,\ldots , 2n \}$, i.e., over all
the $\pi$ permutations of the $2n$ elements
which satisfy $\pi(2l-1)< \pi(2l)$ and $\pi(2l-1) < \pi(2l+1)$.
Hence, if we assume that $x_i \ne x_j$ when $i \ne j$, then
$\omega(m_{x_1}\ldots m_{x_2n})$ is  simply the 
{\it Pfaffian} of the $2n \times 2n$ 
antisymmetric matrix $A_{i,j}:=\omega(m_{x_{i}}m_{x_{j}})$.

An automorphism $\alpha: \mathcal{F} \to \mathcal{F}$ that maps any  
majorana operator onto a linear combination of 
majorana operators is called a {\it Bogoliubov automorphism} \cite{araki}.
For any quasifree state $\omega$, the Bogoliubov transformed
state $\omega':=\omega \circ \alpha$ will 
again be quasifree.

A quasifree state $\omega$ is translation-invariant, i.e. 
$\omega \circ \trans{\mathcal{F}}=\omega $, iff
$\omega(m_{x}m_{y})=\omega(m_{x+2}m_{y+2})$ for all $x,y \in \mathbb{Z}$.
Translation-invariant quasifree states are characterized by 
a majorana two-point matrix which is a $2 \times 2$-block Toeplitz matrix
of the form (for a proof, see e.g. \cite{AB} )
\begin{equation}
\begin{bmatrix}
                      \omega(m_{2x}m_{2y}) & \omega(m_{2x}m_{2y+1})   \\ 
                      \omega(m_{2x+1}m_{2y}) & \omega(m_{2x+1}m_{2y+1})
  \end{bmatrix} 
= 
\frac{1}{2\pi} \int_{-\pi}^{\pi}
\begin{bmatrix}
                      q_{1,1}^{(\omega)}(p) & q_{1,2}^{(\omega)}(p)   \\ 
                      q_{2,1}^{(\omega)}(p) & q_{2,2}^{(\omega)}(p)
  \end{bmatrix}e^{-ip(x-y)} {\rm d}p, \label{TrInv}  
\end{equation}
where $q_{i,j}^{(\omega)}$ are $L^{\infty}([-\pi,\pi])$ functions,
and the matrix function 
\begin{equation*}
Q^{(\omega)}(p)= 
\begin{bmatrix}
q_{1,1}^{(\omega)}(p) & q_{1,2}^{(\omega)}(p)\\
q_{2,1}^{(\omega)}(p) & q_{2,2}^{(\omega)}(p)
\end{bmatrix}
\in L^{\infty}_{2 \times 2}([-\pi,\pi])
\end{equation*}
satisfies
\begin{equation}
(Q^{(\omega)}(p))^{\dagger}=Q^{(\omega)}(p), 
\; \; Q^{(\omega)}(-p)=2 \cdot \mathbbm{1} -(Q^{(\omega)}(p))^T, \; \;
 0 \le Q^{(\omega)}(p) \le 2 \cdot \mathbbm{1} \label{QProp}
\end{equation}
almost everywhere (here $(Q^{(\omega)}(p))^{T}$ denotes the transpose of 
$Q^{(\omega)}(p)$). A translation-invariant 
quasifree state $\omega$ is pure iff 
for almost every $p$ the eigenvalues of
$Q^{(\omega)}(p)$ are either $0$ or $2$. 
$Q^{(\omega)}(p)$ is called the {\it symbol} 
of the majorana two-point matrix $M^{\omega}_{x,y}=\omega(m_{x}m_{y})$.

Now we are ready to characterize the 
translation-invariant quasifree states that are stationary with 
respect to the time-evolution 
 $\glider_{\mathcal{F}}$
(defined in the previous subsection).

\begin{prop}
  \label{prop:qfinv}
A translation-invariant quasifree state $\omega$ is 
invariant under the $\glider_{\mathcal{F}}$ automorphism iff the symbol of its 
majorana two-point matrix 
has the following form:
  \begin{equation*}
   Q^{(\omega)}(p)=\begin{bmatrix}
                      q_{1,1}^{(\omega)}(p) & 0   \\ 
                      0 & q_{2,2}^{(\omega)}(p)
  \end{bmatrix},
  \end{equation*}
where $q_{1,1}^{(\omega)}$ and $q_{2,2}^{(\omega)}$ 
are real $L^{\infty}([-\pi,\pi])$  
functions that take values between $0$ and $2$ (almost everywhere).
\begin{proof}
The $\glider_{\mathcal{F}}$ automorphism acts on the majorana fermions as
\begin{equation*}
\glider_{\mathcal{F}}(m_{2x})=m_{2x-2} \, , \; \; \; 
\glider_{\mathcal{F}}(m_{2x+1})=m_{2x+3} \;,
\end{equation*}
hence it is a Bogoliubov automorphism. Moreover, 
$\glider_{\mathcal{F}}$ commutes 
with the translations. Thus 
$\omega'=\omega \circ \glider_{\mathcal{F}}$
will again be a translation-invariant quasifree state,
and $\omega'$ is equal to $\omega$ iff its 
majorana two-point matrix is the same as that of 
$\omega$. The majorana two-point function of $\omega'$ is
\begin{eqnarray*}
\begin{bmatrix}
                      \omega'(m_{2x}m_{2y}) & \omega'(m_{2x}m_{2y+1})   \\ 
                      \omega'(m_{2x+1}m_{2y}) & \omega'(m_{2x+1}m_{2y+1})
  \end{bmatrix} 
&=&
\begin{bmatrix}
\omega(\glider_{\mathcal{F}}(m_{2x}m_{2y})) & 
\omega(\glider_{\mathcal{F}}(m_{2x}m_{2y+1}))   \\ 
\omega(\glider_{\mathcal{F}}(m_{2x+1}m_{2y})) & 
\omega(\glider_{\mathcal{F}}(m_{2x+1}m_{2y+1}))
  \end{bmatrix}\\
&=&
\begin{bmatrix}
                      \omega(m_{2x-2}m_{2y-2}) & \omega(m_{2x-2}m_{2y+3})   \\ 
                      \omega(m_{2x+3}m_{2y-2}) & \omega(m_{2x+3}m_{2y+3})
  \end{bmatrix}\\
&=&
\begin{bmatrix}
                      \omega(m_{2x}m_{2y}) & \omega(m_{2x-2}m_{2y+3})   \\ 
                      \omega(m_{2x+3}m_{2y-2}) & \omega(m_{2x+1}m_{2y+1})
  \end{bmatrix},
\end{eqnarray*} 
where we have used that $\omega$ is $\trans{\mathcal{F}}$-invariant,
and that $\trans{\mathcal{F}}(m_x)=m_{x+2}$, which follows from
the definition of the majorana operators (\ref{majorana}). 
Comparing the majorana two-point functions,
one can see that $\omega'=\omega$ iff 
$\omega(m_{2x}m_{2y+1})=\omega(m_{2x-2}m_{2y+3})$ for any 
$x,y \in \mathbb{Z}$. From form (\ref{TrInv}) of 
the majorana two-point functions of translation-invariant
quasifree states we know that there exists a $q^{(\omega)}_{1,2}$ 
$L^{\infty}([-\pi,\pi])$ function such that
\begin{equation*}
\omega(m_{2x}m_{2y+1})=\frac{1}{2\pi}\int\limits_{-\pi}^{\pi}
q_{1,2}^{(\omega)}(p) e^{-ip(x-y)} {\rm d}p. 
\end{equation*}
The condition $\omega(m_{2x}m_{2y+1})=\omega(m_{2x-2}m_{2y+3})$, which should hold
for all $x,y \in \mathbb{Z}$, means that 
\begin{eqnarray*}
&&\frac{1}{2\pi}\int\limits_{-\pi}^{\pi}
q_{1,2}^{(\omega)}(p)e^{-ip(x-y)}{\rm d}p -
\frac{1}{2\pi} \int\limits_{-\pi}^{\pi} q_{1,2}^{(\omega)}(p)
e^{-ip((x-1)-(y+1))}{\rm d}p\\
&=&
\frac{1}{2 \pi}\int\limits_{-\pi}^{\pi}q_{1,2}^{(\omega)}(p)(1-e^{2ip})
 e^{-ip(x-y)} {\rm d}p =0
\end{eqnarray*}
must be satisfied. 
It is a well known theorem in functional analysis, that 
an $L^{\infty}([-\pi,\pi])$
function for which the Fourier transformation vanishes  
must be almost everywhere zero, 
hence $q_{1,2}^{(\omega)}(p) \cdot (1-e^{2ip})=0$, from which one concludes 
that $q_{1,2}^{(\omega)}(p)=0$. This means that for a $\trans{\mathcal{F}}$-
and $\glider_{\mathcal{F}}$-invariant state $\omega$ the symbol of the two-point 
majorana matrix has to be of the form
\begin{equation*}
\begin{bmatrix}
                      q_{1,1}^{(\omega)}(p) & 0   \\ 
                      0 & q_{2,2}^{(\omega)}(p)
  \end{bmatrix},
\end{equation*}
where $q^{(\omega)}_{1,1}$ and $q_{2,2}^{(\omega)}$ are $L^{\infty}([-\pi, \pi])$
functions, and according to (\ref{QProp}) 
they also have to satisfy the inequalities 
$0 \le q^{(\omega)}_{1,1}(p) \le 2$ and $0 \le q^{(\omega)}_{2,2}(p) \le 2$
almost everywhere. Thus we have arrived at our proposition.
\end{proof}
\end{prop}

Using the Araki-Jordan-Wigner transformation we can transfer 
such a $\trans{\mathcal{F}}$- and $\glider_{\mathcal{F}}$-invariant quasifree 
state $\omega$ on $\mathcal{F}$
to a state $\omega^{JW}$ on the spin chain which is $\trans{}$-invariant 
and stationary with respect to the glider CQCA $\glider$. 
%%%%%%%%%%%%
\subsubsection{Convergence of quasifree states}
\label{sec:conv_qf}
In this section we will show that under the
repeated action of the  $\glider_{\mathcal{F}}$ automorphism,
any translation-invariant quasifree state will 
converge to one of the 
$\glider_{\mathcal{F}}$-invariant states
specified in Proposition \ref{prop:qfinv}.

\begin{prop} \label{prop:qfconv}
Let $\omega: \mathcal{F} \to \mathbb{C}$ be a translation-invariant 
quasifree state with a majorana two-point matrix that belongs to the symbol
\begin{equation*}
Q^{(\omega)}(p)=
\begin{bmatrix}
q_{1,1}^{(\omega)}(p) & q_{1,2}^{(\omega)}(p)\\
q_{2,1}^{(\omega)}(p) & q_{2,2}^{(\omega)}(p)
\end{bmatrix},
\end{equation*} 
where $Q^{(\omega)} \in L^{\infty}_{2 \times 2}([-\pi,\pi])$ 
satisfies the relations (\ref{QProp}).
The
$n$ time-step evolved state is denoted by
$\omega_n:=\omega \circ \glider_{\mathcal{F}}^n$. The 
pointwise limit
$\omega_{\infty}(A):=\lim_{n \to \infty} \omega_n(A)$ exists 
for all $A \in \mathcal{F}$, and the function
$\omega_{\infty}: \mathcal{F} \to \mathbb{C}$ 
defined in this way will be a translation-invariant quasifree state with a
majorana two-point matrix that belongs to the following symbol
\begin{equation*}
Q^{\left(\omega_{\infty}\right)}(p)=
\begin{bmatrix}
q_{1,1}^{(\omega)}(p) & 0\\
0 & q_{2,2}^{(\omega)}(p)
\end{bmatrix}.
\end{equation*}
\begin{proof}
In the proof of Proposition \ref{prop:qfinv} we showed that
if $\omega$ is a translation-invariant quasifree state, 
then $\omega_2=\omega \circ \glider_{\mathcal{F}}$ will also be
such a state. By induction it follows that 
$\omega_n$ is also a translation-invariant quasifree state
for any $n \in \mathbb{N}^{+}$. 
Hence for an arbitrary odd monomial of majorana operators
$\omega_n(m_{x_1}m_{x_2}\cdots m_{x_{2N+1}})=0$, and the limit 
$\omega(m_{x_1}m_{x_2}\cdots m_{x_{2N+1}}):=
\lim_{n \to \infty}\omega_n(m_{x_1}m_{x_2}\cdots m_{x_{2N+1}})$
exists and is zero.

Next, we prove the pointwise convergence of the majorana 
two-point matrix. The majorana two-point matrix of $\omega_n$ is
\begin{eqnarray*}
\begin{bmatrix}
                      \omega_n(m_{2x}m_{2y}) & \omega_n(m_{2x}m_{2y+1})   \\ 
                      \omega_n(m_{2x+1}m_{2y}) & \omega_n(m_{2x+1}m_{2y+1})
  \end{bmatrix} 
&=&
\begin{bmatrix}
\omega(\glider_{\mathcal{F}}^n(m_{2x}m_{2y})) & 
\omega(\glider_{\mathcal{F}}^n(m_{2x}m_{2y+1}))   \\ 
\omega(\glider_{\mathcal{F}}^n(m_{2x+1}m_{2y})) & 
\omega(\glider_{\mathcal{F}}^n(m_{2x+1}m_{2y+1}))
  \end{bmatrix} \\
&=&
\begin{bmatrix}
                      \omega(m_{2x-2n}m_{2y-2n}) & 
\omega(m_{2x-2n}m_{2y+1+2n})   \\ 
                      \omega(m_{2x+1+2n}m_{2y-2n}) & 
\omega(m_{2x+1+2n}m_{2y+1+2n})
  \end{bmatrix}  \\
&=&
\begin{bmatrix}
     \omega(m_{2x}m_{2y}) & \omega(m_{2x-2n}m_{2y+1+2n})   \\ 
     \omega(m_{2x+1+2n}m_{2y-2n}) & \omega(m_{2x+1}m_{2y+1})
  \end{bmatrix}, 
\end{eqnarray*} 
this means that the symbol $Q^{\left(\omega_n \right)}$ is given by
\begin{equation*}
Q^{\left(\omega_{n}\right)}(p)=
\begin{bmatrix}
q_{1,1}^{(\omega)}(p) & q_{1,2}^{(\omega)}(p)e^{2inp}\\
q_{2,1}^{(\omega)}(p)e^{-2inp} & q_{2,2}^{(\omega)}(p)
\end{bmatrix}.
\end{equation*}
The limits of elements of the two-point majorana matrix are:
\begin{eqnarray}
&&
\begin{bmatrix}
\omega_{\infty}(m_{2x}m_{2y}) & \omega_{\infty}(m_{2x}m_{2y+1})   \\ 
\omega_{\infty}(m_{2x+1}m_{2y}) & \omega_{\infty}(m_{2x+1}m_{2y+1})
  \end{bmatrix}\nonumber\\
&=&
\begin{bmatrix}
\lim\limits_{n \to \infty} \omega_{n}(m_{2x}m_{2y}) & 
\lim\limits_{n \to \infty} \omega_{n}(m_{2x}m_{2y+1})   \\ 
\lim\limits_{n \to \infty} \omega_{n}(m_{2x+1}m_{2y}) 
&\lim\limits_{n \to \infty} \omega_{n}(m_{2x+1}m_{2y+1})
  \end{bmatrix} \nonumber \\
&=&\frac{1}{2 \pi}
\begin{bmatrix}
\lim\limits_{n \to \infty} \int_{-\pi}^{\pi} q_{1,1}^{(\omega)}(p) e^{-i(x-y)p} {\rm d}p & 
\lim\limits_{n \to \infty} \int_{-\pi}^{\pi} q_{1,2}^{(\omega)}(p)
e^{-i(x-y-2n)p} {\rm d}p \\ 
\lim\limits_{n \to \infty} \int_{-\pi}^{\pi} q_{2,1}^{(\omega)}(p)
e^{-i(x-y+2n)p} {\rm d}p 
&\lim\limits_{n \to \infty} \int_{-\pi}^{\pi} q_{2,2}^{(\omega)}(p)  
e^{-i(x-y)p} {\rm d}p
\end{bmatrix}\nonumber \\
&=&\frac{1}{2\pi} \int_{-\pi}^{\pi}
\begin{bmatrix}
                      q_{1,1}^{(\omega)}(p) & 0   \\ 
                      0 & q_{2,2}^{(\omega)}(p)
  \end{bmatrix}e^{-ip(x-y)} {\rm d}p, \label{QfConvQ} 
\end{eqnarray}
where we have used the Riemann-Lebesgue lemma, which states that for
any integrable function $f$ defined on the interval $[a,b]$:
\begin{equation*}
\lim_{z \to \pm \infty} \int\limits_a^b f(x) e^{izx} dx =0.
\end{equation*}

For an arbitrary even monomial of majorana operators the
convergence can be proved by:
\begin{eqnarray}
\omega_{\infty}(m_{x_1}m_{x_2}\cdots m_{x_{2N+1}})&=&
\lim_{n \to \infty} \omega_n(m_{x_1}m_{x_2}\cdots m_{x_{2N+1}})
\nonumber \\
&=&\lim_{n \to \infty} \sum\limits_{\pi}{\text{sgn}}(\pi) 
\prod\limits_{l=1}^{N}
\omega_n(m_{x_{\pi(2l-1)}}m_{x_{\pi(2l)}})=\nonumber \\
&=& \sum\limits_{\pi}{\text{sgn}}(\pi) 
\prod\limits_{l=1}^{N}
\left[\lim_{n \to \infty} \omega_n(m_{x_{\pi(2l-1)}}m_{x_{\pi(2l)}})\right]=
\nonumber \\
&=&\prod\limits_{l=1}^{N}
\omega_{\infty}(m_{x_{\pi(2l-1)}}m_{x_{\pi(2l)}}). \label{QfConvPfaffian}
\end{eqnarray}

Since the finite linear combinations of majorana operators form a 
norm-dense subset in $\mathcal{F}$ 
the limit $\lim_{n \to \infty} \omega_n(f)$ must exist for any 
$f \in \mathcal{F}$ due to the uniform boundedness of the
states in the sequence, and hence we have obtained
with this pointwise limit a 
linear functional $\omega_{\infty}: \mathcal{F} \to \mathbb{R}$.
Moreover, $\omega_{\infty}$ is uniquely determined by the values it takes
on monomials of majorana operators, and since the
equations (\ref{QfConvQ}), (\ref{QfConvPfaffian}) are satisfied 
$\omega_{\infty}$ can only     
be the quasifree state for which the symbol of the
majorana two-point function is given by (\ref{QfConvQ}) .
\end{proof}
\end{prop}

It is worth mentioning that although a pure quasifree state will, of course, stay pure for any time $t$, taking the weak, i.e. pointwise, limit of the states when $t \to \infty$ one obtains a mixed state (if the original state was not $\glider$-invariant). In quantum many body physics such relaxation from a pure to mixed states has been studied in a quite different setting, namely how certain time averages of pure states that evolve under a Hamiltonian dynamics can be described by mixed states \cite{CDEO}. It would be interesting to look at such weak limits when $t \to \infty$ also for Hamiltonian evolution (for systems with infinite degrees of freedom), and to look at the connection between these two approaches.

Finally we have to point out, that the results of this section using the AJW transformation can generally not be transferred to other CQCAs than the glider CQCA $\glider$, because most CQCAs don't have the form of Bogoliubov transformations on the CAR-Algebra. Automata with neighborhoods larger than nearest neighbors can map creation and annihilation operators to products of these on the CAR-algebra and are thus in general not Bogoliubov transformations. (An obvious exception are powers of $\glider$% and shear transformations
.) All automata, that don't leave at least one Pauli matrix locally invariant don't allow for a tail element that is left invariant under the transformed CQCA. These automata are characterized by a constant in their trace polynomial. In conclusion, only nearest neighbor automata without a constant on the trace, i.e. some glider and period two automata and their powers, can be transferred to Bogoliubov transformations on the CAR-algebra. For more details see \cite{Uphoff2008}. 
%%%%%%%%%%
\section{Entanglement generation}
\label{sec:entanglement}
In this section we investigate the entanglement generation properties of CQCAs. First we derive general bounds for the entanglement generation of arbitrary QCAs on the spin-chain both in a translation-invariant and a non-translation-invariant setting. Then we investigate CQCAs acting on stabilizer and quasifree states. We find that the entanglement generation is linear in time, similarly to the case of Hamiltonian time evolutions \cite{CC}.
%%%%%%%%%%%%
\subsection{General bounds on the entanglement generation of QCA}
  \label{sec:bounds}
In this section we derive general bounds on the evolution 
of the entanglement of a finite number of consecutive spins 
with the rest of the chain under the action of a localized automorphism 
(e.g. a QCA). We will only consider the case when the whole chain 
is in a pure state. In this case the 
proper measure of entanglement is given by 
the von Neumann entropy
\begin{equation*}
S= - \rm{Tr} \, \rho_S \log_2 \rho_S,
\end{equation*} 
where $\rho_S$ is the reduced density matrix of the
finite segment of spins.\footnote{Our bounds for the
entropy generation hold also for non-pure states, although in
this case the von Neumann entropy is not directly related to
entanglement.}
%%%%%%%%%%%%%
 \subsubsection{The non-translation-invariant case}
\begin{thm}
  Consider the observable algebra of an 
  infinite chain of $d$-level systems
  \begin{equation*}
    \mathfrak{A}^{(d)}:= \bigotimes^{+\infty}_{i=-\infty}\mathfrak{A}_{i}^{(d)},
    \; \; \mathfrak{A}_{i}^{(d)} \cong M_d,
  \end{equation*} 
  and an automorphism 
$T: \mathfrak{A}^{(d)} \to  \mathfrak{A}^{(d)}$. 
Let us introduce the notation:
  $\mathfrak{A}_{[m_1,m_2]}^{(d)}:=
\bigotimes^{m_2}_{k=m_1}\mathfrak{A}_{k}^{(d)}$.
Suppose that $T$ satisfies the locality
condition
  \begin{equation*}
    T(\mathfrak{A}_{[k_1,k_2]}^{(d)}) 
\subset \mathfrak{A}_{[l_1,l_2]}^{(d)} \;, \label{local}
  \end{equation*}  
 for some fixed integers $k_1,k_2,l_1,l_2$, with
 $l_1 \le k_1$ and $k_2 \le l_2$.

 Let $\omega$ be a state on the spin-chain $\mathfrak{A}^{(d)}$,
and let us define the $T$-evolved state 
 as $\omega':=\omega \circ T$. The restrictions of these 
states to a subalgebra $\mathfrak{A}_{[m_1,m_2]}^{(d)}$
will be denoted by $\omega_{[m_1,m_2]}$ and $\omega'_{[m_1,m_2]}$,
  respectively. Then the following bounds hold
for the von Neumann entropies of the restricted states:
  \begin{equation}
    S(\omega_{[k_1,k_2]})-2 n \log_2 d 
    \le S(\omega'_{[k_1,k_2]}) 
    \le S(\omega_{[k_1,k_2]})+2n \log_2 d, \label{bound2}
  \end{equation}
where $n=l_2-l_1-k_2+k_1$. Moreover, these bounds are sharp. 
  \begin{proof}
    Restricting the automorphism $T$ to the subsubalgebra  
$\mathfrak{A}_{[k_1,k_2]}^{(d)}$, 
    we get a monomorphism\footnote{An 
      injective but not (necessarily) surjective homomorphism.} 
 $T_{[k_1,k_2]}:\mathfrak{A}_{[k_1,k_2]}^{(d)} \to \mathfrak{A}_{[l_1,l_2]}^{(d)}$.
This monomorphism can be extended 
to an automorphism $\widetilde{T}:
\mathfrak{A}_{[l_1,l_2]}^{(d)} \to 
\mathfrak{A}_{[l_1,l_2]}^{(d)}$.\footnote{This can be 
      simply seen by noting that 
      $\mathfrak{A}_{[l_1,l_2]}^{(d)}=T(\mathfrak{A}_{[k_1,k_2]}^{(d)}) 
      \otimes 
      \mathfrak{B}$,
      where $\mathfrak{B} \cong (\bigotimes_{k=l_1}^{k_1-1}
\mathfrak{A}_{k}^{(d)})\otimes
      (\bigotimes_{j=k_2+1}^{l_2}\mathfrak{A}_{j}^{(d)})$, 
if $Q$ is an isomorphism
      between the latter two algebras, then one can define
      $\widetilde{T}_{[l_1,l_2]}$ as $T_{[k_1,k_2]} \otimes Q$.}
Let us introduce the following state on $\mathfrak{A}_{[l_1,l_2]}$:
    \begin{equation*}
      \widetilde{\omega}'([l_1,l_2]):=\omega_{[l_1,l_2]} 
      \circ \widetilde{T}.
    \end{equation*}
From this construction it immediately follows that
$\widetilde{\omega}'([l_1,l_2])_{[k_1,k_2]}=\omega'_{[k_1,k_2]}$.
Moreover, since 
$\widetilde{\omega'}([l_1,l_2]))$ and $\omega_{[l_1,l_2]}$ are connected by
an automorphism their von Neumann entropies are equal:
$S(\widetilde{\omega}'([l_1,l_2]))=S(\omega_{[l_1,l_2]})$.

    We will prove the bounds (\ref{bound2}) using the subadditivity of the von
    Neumann entropy. The subchain
    $\mathfrak{A}_{[l_1,l_2]}^{(d)}$ can be divided as 
    $\mathfrak{A}_{[l_1,l_2]}^{(d)}=\mathfrak{A}_{[k_1,k_2]}^{(d)} 
\otimes \mathfrak{A}_{rest}^{(d)}$,
    where $\mathfrak{A}_{rest}^{(d)}$ is isomorphic to the algebra of
    $d^{n} \times d^{n}$ matrices,
    hence the maximal entropy of a state defined on $\mathfrak{A}_{rest}^{(d)}$
    is $n \log_2 d$. The triangle inequality and the subadditivity theorem give the following 
    inequalities: 
    \begin{eqnarray*}
      S(\omega_{[k_1,k_2]})-n \log_2 d \le &S(\omega_{[l_1,l_2]})& \le
      S(\omega_{[k_1,k_2]})+n \log_2 d\\
      S(\omega'_{[k_1,k_2]})-n \log_2 d \le &S
      (\widetilde{\omega}'({[l_1,l_2]}))& \le
      S(\omega'_{[k_1,k_2]})+ n \log_2 d
    \end{eqnarray*}
    Now, using that $S(\omega_{[l_1,l_2]})=S(\widetilde{\omega}'({[l_1,l_2]}))$
    we immediately obtain the bounds (\ref{bound2}).
    
    The sharpness of the inequalities follows if we consider
    a state on the spin-chain where the sites at $2i$ are maximally entangled 
    with the sites at $2i+1$ and we consider the translation $\tau$ which just 
    shifts all one-cell algebras by one cell to the right as our 
    time-evolution.
    Then $k_2-k_1=2$ and $l_2-l_1=3$, and we get $n=l_2-l_1-k_2+k_1=1$.
    Now, restricting this state to the subalgebra $\mathfrak{A}_{2i, 2i+3}^{(d)}$
    the entropy of the restriction is zero. However, the entropy of this
    restriction after the time evolution will be $2\log_2 d$, 
     since the two sites at the border will
    be maximally entangled with sites outside the considered region. 
    Let us note, that in some sense we broke the translation-invariance 
    only minimally,
    since the considered state is invariant under the 
    square of the translations $\tau^2$.
        
    In the above example the generated entanglement is destroyed 
    in the next step, so the bound is only saturated for one time step. 
    But a slightly more involved example shows, that the bound can be 
    saturated for arbitrarily many timesteps:
    
    Let again $T$ be the translation automorphism on $\mathfrak{A}^{(d)}$, 
    and let us consider the state  
    on a spin-chain which is defined as the direct 
    product of totally mixed states between the lattice site at $i$ and the 
    lattice site at $-i+1$ for all $i$. So in this state the lattice 
    site at $1$ is fully entangled with the lattice site at $0$, the 
    lattice site at $2$ is fully entangled with the lattice site at $-1$, 
    and so on. Now, we will consider subsystems of arbitrary 
    length $k=k_2-k_1$. If $k$ is even, $k=2j$, then the subsystem 
    we consider is the interval $[-j+1,j]$. Its original entropy is $0$, 
    and the entropy grows linearly during the time-evolution saturating 
    our linear bound until it reaches the maximal entropy it can obtain, 
    namely $k\log_2 d$. After this it stays constant. 
    If $k$ is odd, $k=2j+1$, we consider the interval $[-j,j]$ 
    as our subsystem. The original entropy of the subsystem 
    is $\log_2 d$, and the entropy grows linearly and saturates our bound 
    until it reaches $(2j+1) \log_2 d$, then it stays constant.
  \end{proof}
\end{thm}
  \subsubsection{The translation-invariant case}
  In the previous subsection it was shown that the 
bounds (\ref{bound2}) on entanglement generation
are sharp in the general case. However, 
considering translation-invariant states
and QCA automorphism, i.e., automorphisms that
commute with the translations, we can sharpen  
these bounds further.

  \begin{thm}
    Consider the observable algebra of an 
    infinite chain of $d$-level systems
    \begin{equation*}
      \mathfrak{A}^{(d)}:= \bigotimes^{+\infty}_{i=-\infty}\mathfrak{A}_{i}^{(d)},
      \; \; \mathfrak{A}_{i}^{(d)} \cong M_d, 
    \end{equation*}
    and a QCA automorphism $T$ acting on $\mathfrak{A}^{(d)}$ 
    having a neighborhood of $n$ ``extra cells'', i.e., $T$ is an
    automorphism satisfying
    \begin{equation}
      T(\mathfrak{A}_i)^{(d)} \subset \mathfrak{A}_{i-n_1}^{(d)} \otimes 
      \mathfrak{A}_{i-n_1+1}^{(d)} 
      \otimes \dots \otimes \mathfrak{A}_{i+n_2-1}^{(d)} \otimes
      \mathfrak{A}_{i+n_2}^{(d)} \; , \; \; \; \; T \circ \tau =\tau \circ T,    
     \label{loc}
    \end{equation}
    where $i \in \mathbb{Z}$, $\tau$ is the translation automorphism on 
    $\mathfrak{A}^{(d)}$, and $n_1$ and $n_2$ are integers 
    satisfying $n_1+n_2=n \ge 0$.

    Let $\omega$ be a translation-invariant state 
    on the spin-chain, and let us define 
    the time-evolved state (at time $t\in \mathbb{N}$) 
    as $\omega(t):=\omega \circ T^{t}$.
    The von Neumann entropy $S_L(t)$ of the restriction of
    $\omega(t)$ to $L$ consecutive qubits can be bounded in the following
    way:
    \begin{equation}
      S_L(0)-nt \log_2 d \le S_L(t) \le S_L(0)+nt \log_2 d. \label{bound}
    \end{equation}
    Moreover, these bounds are sharp for $d=2$.
    
    \begin{proof}
      Since the state $\omega$ is translation invariant and $T$ commutes 
      with the translations, the ``entropy production'' is the same
      for the automorphisms $T$ and $T \circ \tau^k$ ($k \in \mathbb{Z}$),
      hence we can assume without loss of generality
      that in Eq. (\ref{loc})  $n_1, n_2 \ge 0$.

      Let us denote the
      restriction of a state
      $\varphi:\mathfrak{A}^{(d)} \to \mathbb{C}$ to 
      $\mathfrak{A}_{[m_1,m_2]}^{(d)}=\bigotimes^{m_2}_{k=m_1}
      \mathfrak{A}_{k}^{(d)}$
      by $\varphi_{[m_1,m_2]}$.
      Consider the subalgebra $\mathfrak{A}_{[0,L-1]}^{(d)}$
      of $\mathfrak{A}^{(d)}$, which corresponds to $L$ qubits.
      Restricting the automorphism $T$ 
      to this subalgebra, we get a monomorphism 
      $T_{L}:\mathfrak{A}_{[0,L-1]}^{(d)} \to 
      \mathfrak{A}_{[-n_1,L-1+n_2]}^{(d)}$.
      We will also consider the inverse automorphism
      $T^{-1}$, and restrict $T^{-1}$ to a monomorphism
      $(T^{-1})_L:\mathfrak{A}_{[0,L-1]}^{(d)} \to 
      \mathfrak{A}_{[-n_2,L-1+n_1]}^{(d)}$.\footnote{The fact that the range of $\mathfrak{A}_{[0,L-1]}$ 
          under the the action of $T^{-1}$ is in 
          $\mathfrak{A}_{[-n_2,L-1+n_1]}$ was shown in \cite{Vogts2009}.}.
      The monomorphism $(T^{-1})_L$ can be extended 
      to an automorphism $\widetilde{T}^{-1}_L:
      \mathfrak{A}_{[-n_2,L-1+n_1]}^{(d)} \to 
      \mathfrak{A}_{[-n_2,L-1+n_1]}^{(d)}$.

      Let $\omega$ be a translation-invariant state on $\mathfrak{A}^{(d)}$, 
       then 
      $\omega(1):=\omega \circ T$ will be translation-invariant, too.
      Let us also define the following state on 
      $\mathfrak{A}_{[-n_2,L-1+n_1]}^{(d)}$
      \begin{equation*}
        \widetilde{\omega}([-n_2,L-1+n_1]):=\omega(1)_{[-n_2,L-1+n_1]} 
        \circ \widetilde{T}^{-1}_L.
      \end{equation*} 
      The von Neumann entropy of 
      $\widetilde{\omega}([-n_2,L-1+n_1])$ and 
      $\omega(1)_{[-n_2,L-1+n_1]}$ are the same
      (since they are connected by an automorphism), and it follows from the definition of $\widetilde{T}^{-1}_L$ that
      $\widetilde{\omega}([-n_2,L-1+n_1])_{[0,L-1]}=\omega_{[0,L-1]}$.

      Now, from the strong subadditivity of the 
      von Neumann entropy it follows that for a
      translation-invariant state $\omega$
      $S(\omega_{[m_1,m_2]}) \ge S(\omega_{[k_1,k_2]})$ if  $m_2-m_1 \ge k_2-k_1$
      \cite{fannesbuch}, hence
      \begin{equation}
        S_L(1)=S(\omega(1)_{[0,L-1]}) \le S(\omega(1)_{[-n_2,L-1+n_1]}). \label{FirstIneq}
      \end{equation}
      On the other hand, using the subadditivity of the entropy
      for the state $\widetilde{\omega}([-n_2,L-1+n_1])$ (dividing the
      observable algebra of the subchain $[-n_2,n_1]$ as:
      $\mathfrak{A}_{[-n_2,L-1+n_1]}^{(d)}=\mathfrak{A}_{[0,L-1]}^{(d)}\otimes 
      (\mathfrak{A}_{[-n_2,-1]}^{(d)} \otimes \mathfrak{A}_{[L,L-1+n_1]}^{(d)}$)),
      we get:
      \begin{equation}
        S_L(\widetilde{\omega}([-n_2,L-1+n_1]))\le 
        S_L(\widetilde{\omega}([-n_2,L-1+n_1])_{[0,L-1]})+(n_1+n_2)\log_2 d=
        S_L(0)+n \log_2 d \label{SecondIneq}
      \end{equation}
      Combing the fact that $S(\omega(1)_{[-n_2,L-1+n_1]})=
      S(\widetilde{\omega}([-n_2,L-1+n_1]))$ 
      with the inequalities (\ref{FirstIneq}) and (\ref{SecondIneq})
      we arrive at the $S(1) \le S(0) +n \log_2 d$.
      By simple induction we obtain the desired 
      upper bound:
      \begin{equation*}
        S_L(t)\le S_L(0) +nt \log_2 d.
      \end{equation*}

      The lower bound in (\ref{bound}) can simply be obtained by
      "reversing the time arrow": suppose that for a QCA automorphism $T$
      this lower bound does not hold, this would mean that for the
      QCA $T^{-1}$ the upper bound wouldn't
      hold, which is a contradiction as we proved the upper bound just now.

      The sharpness of the inequalities for $d=2$ follows from the study of
      of Clifford QCAs acting on the ``all spin up state'' in Section \ref{sec:ent_stab}.
    \end{proof}
  \end{thm}
%%%%%%%%%%%%%%%%%%%%
\subsection{Entanglement generation starting from translation-invariant stabilizer states}
\label{sec:ent_stab}
In this section we will consider the entanglement generation of CQCA acting on translation-invariant pure stabilizer states. We will first calculate the bipartite entanglement in a general translation-invariant pure stabilizer state. Using this result we will present a proof of asymptotically linear growth of entanglement for non-periodic CQCAs.

For every translation-invariant stabilizer state $\omega$ with stabilizer $\stab=\langle\weyl{\hat\tau^x\xi},\,x\in\ZZ\rangle$ there exists a CQCA $\CQCA$ which maps $\weyl{0,1}$ to $\weyl{\xi}$. Each $\stab=\langle\weyl{\hat\tau^x\xi},\,x\in\ZZ\rangle$ defines a unique translation-invariant stabilizer state if and only if $\xi$ is reflection invariant and $\gcd(\xi_+,\xi_-)=1$ \cite{SchlingemannCQCA}. The image of a one-site Pauli matrix under the action of a CQCA $B$ is always of this form. The study of the entanglement generation of CQCAs acting on initially unentangled stabilizer product states is thus equivalent to the study of the entanglement properties of translation-invariant stabilizer states.
 
There are several results on the entanglement entropy for stabilizer states in the literature, the most general example would be the formalism developed in \cite{Fattal2004}. One case considered is a bipartite split of the state $\omega_{AB}$ with respect to the subsystems $A$ and $B$. The set of stabilizers $\stab$ then splits up into three sets. $\stab_A$ and $\stab_B$ are the local stabilizers, which act non-trivially only on part $A$ resp. $B$. The third set $\stab_{AB}$ accounts for correlations between the subsystems. It is defined as follows:
\begin{defi}
  The correlation subgroup $\stab_{AB}$ for a bipartite stabilizer state is generated by all stabilizer generators that have support on both parts of the system.
\end{defi}
$\stab_A$ and $\stab_B$ together form the so called local subgroup. The correlation subgroup $\stab_{AB}$ can be brought into a form where it consists of pairs of stabilizers whose projections on $A$ (and $B$) anticommute, but commute with all elements of other pairs and the local subgroup. 
The entanglement or von Neumann entropy of such a stabilizer state is 
\begin{equation}
  \label{eq:stab_entropy}
  E(\omega_{AB})=\frac{1}{2}|\stab_{AB}|,
\end{equation}
if $\omega_{AB}$ is a pure state, where $|\stab_{AB}|$ is the size of the minimal generating set of $\stab_{AB}$. 

Unfortunately, in \cite{Fattal2004} only finitely many qubits are considered. The proof of (\ref{eq:stab_entropy}) relies heavily on this fact. However, there is a different approach to the bipartite entanglement in stabilizer states which we use to extend this result to infinitely many qubits.

In our approach we make use of the phase space description of stabilizer states introduced in \cite{SchlingemannCQCA}. A stabilizer state is fixed by a set of defining commuting Pauli products. In the phase space description commutation relations are encoded in the symplectic form $\sigma(\xi,\eta)$. If $[\weyl{\xi},\weyl{\eta}]=0$ we have $\sigma(\xi,\eta)=0$. Thus abelian algebras of Weyl operators (Pauli products) correspond to subspaces on which the symplectic form vanishes. Those subspaces are called isotropic subspaces. If for an isotropic subspace $\isosub$ $\sigma(\eta,\xi)=0$, $\forall\xi\in\isosub$ implies $\eta\in\isosub$, we call $\isosub$ maximally isotropic. Maximally isotropic subspaces correspond to maximally abelian algebras. In \cite{SchlingemannCQCA} it was shown that the above condition on $\xi$ (reflection invariance and $\gcd(\xi_+,\xi_-)=1$) is equivalent to the condition that $\fring\xi$ is a maximally isotropic subspace. $\fring\xi$ denotes the space generated by the products of $\xi$ and all elements of $\fring$. We have $\langle\weyl{\eta},\eta\in\fring\xi\rangle=\langle\weyl{\hat\tau^x\xi},\,x\in\ZZ\rangle=\stab$. $\fring\xi$ is the phase space of the stabilizer group $\stab$. By $(\fring\xi)_A$ etc. we denote the phase spaces of $\stab_A$ etc.

\begin{thm}
  \label{thm:stab_ent}
  The number of maximally entangled qubit pairs in a pure translation invariant stabilizer state stabilized by $\stab=\langle\weyl{\hat\tau^x\xi},\,x\in\ZZ\rangle$ on a bipartite spin-chain is the number of pairs $\eta^i,\,\zeta^i\in(\fring\xi)_{AB}$ with $\sigma_A(\eta^i,\zeta^j)=\delta_{ij}$, $\sigma(\eta^i_A,\eta^j_A)=\sigma(\zeta^i,\zeta^j)=0$, $\sigma(\eta^i_A,\mu)=\sigma(\eta^i,\mu)=0,\,\forall\mu\in(\fring\xi)_A$, and $\sigma(\zeta_A^i,\mu)=\sigma(\zeta^i,\mu)=0,\,\forall\mu\in(\fring\xi)_A$. where $\eta_A$, $\zeta_A$ denote the restriction of the phase space vectors to subsystem $A$ completed with $0$ on $B$ so we can use the the symplectic form $\sigma$ of the whole chain.
  \begin{proof}
    If we restrict the stabilizer to subsystem $A$ (or $B$) it is in general not translation invariant any more. Therefore the corresponding subspace is not maximally isotropic. The restricted state is not a pure translation invariant stabilizer state. However the uncut stabilizer operators in $\stab_A$ stabilize a subspace of the statespace. The elements of the correlation subgroup $\stab_{AB}$ which is generated by the cut stabilizer generators map this subspace onto itself because they commute with the elements of $\stab_A$. From the theory of quantum error correction \cite{Gottesman1997} we know that a pair of operators leaving a stabilized subspace invariant can be used to encode a logical qubit if the operators fulfill the same commutation relations as $\sigma_1$ and $\sigma_3$. As the restrictions of the elements of $\stab_{AB}$ don't have to commute such pairs of operators can exist. In the phase space description this means, that we have to find $\eta,\,\zeta\in(\fring\xi)_{AB}$ with $\sigma(\eta_A,\zeta_A)=1$,  $\sigma(\eta_A,\mu_A)=\sigma(\eta,\mu)=0,\,\forall\mu\in(\fring\xi)_A$, and $\sigma(\zeta_A,\mu_A)=\sigma(\zeta,\mu)=0,\,\forall\mu\in(\fring\xi)_A$. Several such pairs encode several qubits. Of course the operators from different pairs have to commute. Thus the qubits are encoded by pairs of operators $(\bar\sigma_1^i,\bar\sigma_3^i)=(\weyl{\eta^i},\weyl{\zeta^i})$ whose phase space vectors fulfill $\sigma_A(\eta^i,\zeta^j)=\delta_{ij}$, $\sigma(\eta^i_A,\eta^j_A)=\sigma(\zeta^i,\zeta^j)=0$, $\sigma(\eta^i_A,\mu)=\sigma(\eta^i,\mu)=0,\,\forall\mu\in(\fring\xi)_A$, and $\sigma(\zeta_A^i,\mu)=\sigma(\zeta^i,\mu)=0,\,\forall\mu\in(\fring\xi)_A$. 
    As we have $\sigma(\eta_A,\zeta_A)+\sigma(\eta_B,\zeta_B)=\sigma(\eta_A+\eta_B,\zeta_A+\zeta_B)=\sigma(\eta,\zeta)=0,\,\forall\eta,\zeta\in\fring\xi$ we know that $\sigma(\eta_B,\zeta_B)=\sigma(\eta_A,\zeta_A)$, thus the restrictions of our operators to system $A$ and $B$ fulfill the same commutation relations. We therefore have pairs of operators of the form $(\bar\sigma_1^A\otimes\bar\sigma_1^B,\bar\sigma_3^A\otimes\bar\sigma_3^B)$. Each such a pair encodes a Bell pair as seen in Example \ref{exam:bell}. Thus the number of maximally entangled qubit pairs is the number of such pairs of operators. 
  \end{proof}
\end{thm}

We now show that the number of qubit pairs is $\frac{1}{2}|\stab_{AB}|$. As mentioned above, only Weyl operators $\weyl{\xi}$ fulfilling certain conditions can span the stabilizer $\stab=\langle\weyl{\hat\tau^x\xi},\,x\in\ZZ\rangle$ of a pure state. On the level of tensor products of Pauli matrices the above conditions have three important consequences that stem from the requirement for $\xi_+$ and $\xi_-$ to have no common divisors and to be reflection invariant:
\begin{enumerate}
  \item The length of the product has to be odd, because palindromes of even length are always divisible by $(1+\varia)$. We will write $l=2n+1$. 
  \item The central element of the product can not be the identity. Else $\xi$ has the divisor $(u^{-1}+u)$.
  \item At least two different types of elements (both different from the identity) have to occur (e.g. $\sigma_1$ and $\sigma_2$). Else $\xi_+=0$ or $\xi_-=0$ or $\xi_+=\xi_-$, each case implying common divisors.
\end{enumerate}
If we make a bipartite cut\footnote{Due to the translation invariance all possible cuts are equivalent.} in our system, $2n$ stabilizers will be affected. All other operators are localized on one side of the cut, only those with localization on both sides are elements of $\stab_{AB}$. If we could find $k$ pairs of anticommuting operators in the projections of $\stab_{AB}$ on the right (or left) halfchain there would be $k$ pairs of maximally entangled qubits. In fact we can always find $k=n$ such pairs and thus $k=\frac{1}{2}|\stab_{AB}|$.

\begin{defi}
  The bipartite entanglement $E(\omega_\xi)$ of a translation-invariant stabilizer state $\omega_\xi$ is the number of maximally entangled qubit pairs with respect to any bipartite cut. 
\end{defi}

\begin{thm}
  \label{thm:stab_ent_encode}
  A pure translation-invariant stabilizer state $\omega_\xi$ of stabilizer generator length $2n+1$ entangles $n$ qubit pairs maximally with respect to any bipartite cut. 
  \begin{equation}
    E(\omega_\xi)= \frac{\mathrm{length}(\xi)-1}{2}.
  \end{equation}
\end{thm}
For the proof of this theorem we refer to Appendix~\ref{sec:app_stab_ent_encode}.

Now it remains to show how the stabilizer generator length evolves under the action of CQCAs, and to deduce the asymptotic entanglement generation rate. 
\begin{defi}
  The asymptotic entanglement generation rate from stabilizer states for CQCAs is defined as 
  \begin{equation}
    \frac{\Delta E}{\Delta t} =\lim_{t\to\infty}\frac{1}{t}E(\omega_\xi,t),
  \end{equation}
  where $E(\omega_\xi,t)$ is the bipartite entanglement at time~$t$.
\end{defi}

We will now prove the following theorem:
\begin{thm}
  \label{thm:ent_stab}
  The asymptotic entanglement generation per step of a general centered CQCA is the highest exponent in its trace polynomial, $\dg{\tr\sca}$.
\end{thm}
For the proof we need the following lemma:
\begin{lem}
  \label{lem:stab_length_growth}
  The length $2n+1$ of the stabilizer generator of a stabilizer state grows asymptotically with 
  \begin{equation}
    \frac{\Delta n}{\Delta t} =\lim_{t\to\infty}\frac{1}{t}n(t,\xi) = \dg{\tr\sca}
  \end{equation}
  for any centered CQCA $\CQCA$ and any translation invariant pure stabilizer state $\omega_{\xi}$.
  \begin{proof}
    We know that CQCAs map pure stabilizer states to pure stabilizer states and stabilizer generators to stabilizer generators \cite{SchlingemannCQCA}. Thus the image of a pure stabilizer state $\omega_{\xi}$ with stabilizer $\stab=\langle\weyl{\hat\tau^x\xi},\,x\in\ZZ\rangle$ under $\CQCA^t$ is again a pure stabilizer state. We can write $\xi={\bf b}\tbinom{0}{1}$ and $\omega_\xi=\omega_{(0,1)}\circ B$. Therefore we can write the evolved state as $\omega_{\sca^t\xi}=\omega_\xi\circ\CQCA^t=\omega_{(0,1)}\circ\CQCA^t\circ B$. The length of the stabilizer generator is determined by the highest order of the stabilizer generators polynomials, $\dg{\xi}$. Namely the stabilizer generator is of length $2\cdot\dg{\xi}+1$. So we have to calculate $\dg{\sca^t\xi}=\dg{\sca^t{\bf b}\tbinom{0}{1}}$.
    
    An arbitrary product of CSCAs can be written as $\prod_{i=1}^k\sca_i$. The series $(a_l)_{1\le l\le k}=\dg{\prod_{i=1}^l\sca_i}$ is subadditive, i.e. $a_{n+m}\le a_n+a_m$, because concatenation of CSCAs is essentially the multiplication and addition of polynomials which is subadditive in the exponents. For subadditive series $a_n$ Fekete's Lemma \cite{Fekete1923} states that the limit $\lim_{n\to\infty}\frac{a_n}{n}$ exists. In our case the series is always positive, so the limit is positive and finite. To determine the limit, we use the subsequence of the $t=2^k$th steps. Using the Cayley-Hamilton theorem we get
    \begin{equation*}
      \sca^{2^k}=\sca(\tr\sca)^{2^k-1}+\Id\sum_{i=1}^k(\tr\sca)^{2^k-2^i}
    \end{equation*}
    and
    \begin{eqnarray*}
      \mu(k)&\mathrel{\mathop:}=&\dg{\sca^{2^k}\scb\tbinom{0}{1}}\\
      &=&\dg{\sca\scb\tbinom{0}{1}(\tr\sca)^{2^k-1}+\scb\tbinom{0}{1}\sum_{i=1}^k(\tr\sca)^{2^k-2^i}}\\
      &=&\eta\cdot c(k)+\zeta\cdot d(k)
    \end{eqnarray*}
    with $\eta=\sca\scb\tbinom{0}{1}$, $\zeta=\scb\tbinom{0}{1}$, $c(k)=(\tr\sca)^{2^k-1}$, and $d(k)=\sum_{i=1}^k(\tr\sca)^{2^k-2^i}$.
    $c(k)$ and $d(k)$ fulfill the recursion relations
    \begin{eqnarray*}
      c(k+1)&=&(\tr\sca)^{2^{k+1}-1}=(\tr\sca)^{2^k-1}(\tr\sca)^{2^k}=c(k)(\tr\sca)^{2^k},\\
      d(k+1)&=&\sum_{i=1}^{k+1}(\tr\sca)^{2^{k+1}-2^i}=\sum_{i=1}^{k}(\tr\sca)^{2^k+2^k-2^i}+(\tr\sca)^0=d(k)(\tr\sca)^{2^k}+1.
    \end{eqnarray*}
    Therefore 
    \begin{equation*}
      \mu(k+1)=\mu(k)(\tr\sca)^{2^k}+\zeta.
    \end{equation*}
    At this point we need a binary case distinction. Either (1.) $\dg{\mu(k)}$ is uniformly bounded by $\dg{\zeta}$, implying peridocity of $\sca$, or (2.) $\dg{\mu(k)}$ is unbounded and passes $\dg{\zeta}$ so no cancellation can occur and we can easily calculate the limit.
    \begin{enumerate}
      \item $\dg{\mu(k)}$ uniformly bounded by $\dg{\zeta}$ implies $\dg{\mu(k)}\le\dg{\zeta}$ for all $k$. Then $n(t=2^k,\xi)=\dg{\mu(k)}$ is bounded and $\frac{\Delta n}{\Delta t}=0$. But $n(t,\xi)$ bounded also implies $\sca$ periodic and therefore $\dg{\tr\sca}=0$. Thus we have $\frac{\Delta n}{\Delta t}=0=\dg{\tr\sca}$.
      \item If $\dg{\mu(k)}$ is not uniformly bounded by $\dg{\zeta}$, there exists a $\mu(k_0)$ with $\dg{\mu(k_0)}>\dg{\zeta}$. Then $\dg{\mu(k_0+1)}=\dg{\mu(k_0)}+\dg{\tr\sca}(2^{k_0})$ and by recursion $\dg{\mu(k+1)}=\dg{\mu(k)}+\dg{\tr\sca}(2^{k}),\,\forall k\ge k_0$.
      Now we can calculate the limit:
      \begin{eqnarray*}
        \lim_{k\to\infty}\frac{1}{2^k}\dg{\sca^{2^k}{\bf b}\tbinom{0}{1}}&=&\lim_{k\to\infty}\frac{1}{2^k}\dg{\mu(k)}\\
        &=&\lim_{k\to\infty}\frac{1}{2^k}\left(\dg{\mu(k_0)}+\sum_{i=1}^{k-k_0}\dg{\tr\sca}(2^{k-i})\right)\\
        &=&\underbrace{\lim_{k\to\infty}\frac{1}{2^k}\dg{\mu(k_0)}}_{=0}+\lim_{k\to\infty}\sum_{i=1}^{k-k_0}\dg{\tr\sca}\frac{1}{2^{i}}\\
        &=&\dg{\tr\sca}.
      \end{eqnarray*}
    \end{enumerate}
    Thus in all cases $\frac{\Delta n}{\Delta t}=\dg{\tr\sca}$.
  \end{proof}
\end{lem}
Now we can proceed to the proof of Theorem \ref{thm:ent_stab}.
\begin{proof}[Proof of Theorem \ref{thm:ent_stab}]
  As shown in Theorem \ref{thm:stab_ent_encode}, a stabilizer state of stabilizer generator length $2n+1$ encodes $n$ maximally entangled qubits with respect to a bipartite cut. In Lemma \ref{lem:stab_length_growth} we showed that the minimal length of a stabilizer generator grows asymptotically with $2\cdot\dg{\tr\sca}$ under the action of a CQCA $\CQCA$. Together these results prove the theorem.
\end{proof}

Figure \ref{fig:bipartite_stab} illustrates this behavior for different CQCAs.

%\begin{figure}[htbp]
\begin{figure}[ht]
  \begin{center}
    \input{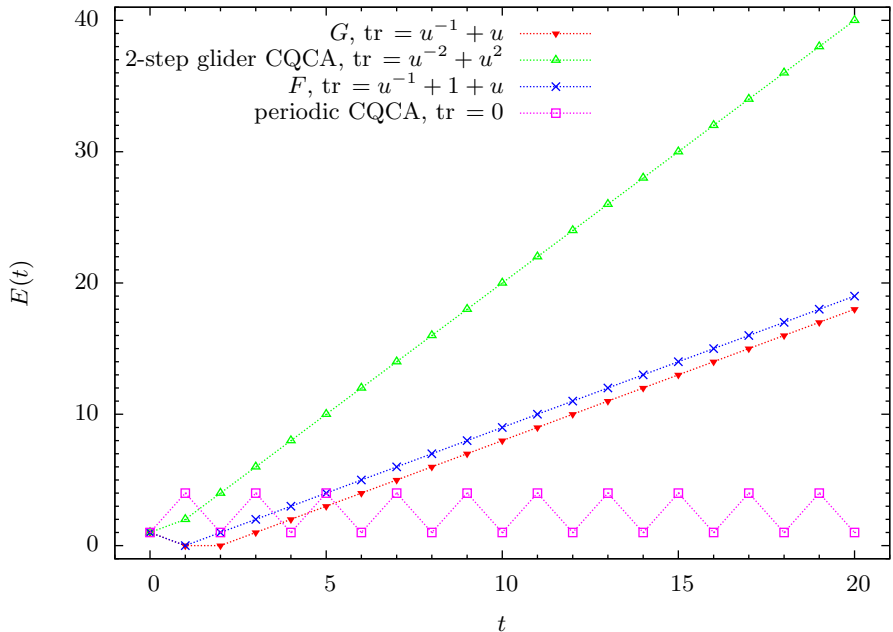}
    \caption{Entanglement generation for the stabilizer state with $\weyl{\xi}=\sigma_2\otimes\sigma_1\otimes\sigma_2$ in a bipartite setting with different CQCAs. One can see that entanglement can also be destroyed, but grows asymptotically linear with the number of timesteps $n$. The coefficient is given by the degree of the trace of the CSCA matrix.}
    \label{fig:bipartite_stab}
  \end{center}
\end{figure}

We can also calculate the entanglement of a finite region, i.e.,\ $L$ consecutive spins, with the rest of the chain. To do this calculation, we use the same method as above, and arrive at the following theorem.
\begin{thm}
  \label{thm:stab_ent_encode_finite}
  Given a pure translation-invariant stabilizer state of stabilizer generator length $2n+1$, a region of length $L$ shares $2n$ maximally entangled qubit pairs with the rest of the chain if $2n\le L$ and $L$ qubits pairs if $2n>L$.
  \begin{proof}
    The proof works exactly as in the bipartite case. In the case $2n\le L$ the cut stabilizers are only cut on one side. But all stabilizers that are cut on the left side commute with those cut on the right side. Thus we have two independent cuts of the bipartite case and therefore $2n$ pairs of maximally entangled qubits. In the case $2n > L$ some stabilizers are cut on both sides. We arrange them in a $(2n+L)\times 2L$-matrix like in the proof of Theorem \ref{thm:stab_ent_encode} and use the same technique to produce the mutually commuting anticommuting pairs which encode the qubits. We always find $L$ pairs of maximally entangled qubits.
  \end{proof}
\end{thm}
For the evolution of entanglement under the action of a CQCA $\CQCA$, this means that starting with a product stabilizer state, the entanglement grows with $2\cdot\tr\sca$ until it reaches $L$. Then it remains constant. If we start with a general translation-invariant stabilizer state, the entanglement might decrease at first. After some time it starts increasing and reaches $L$, where it remains if the CQCA is not periodic. Results are shown in Figure \ref{fig:tripartite_stab}.

%\begin{figure}[htbp]
\begin{figure}[ht]
  \begin{center}
    \input{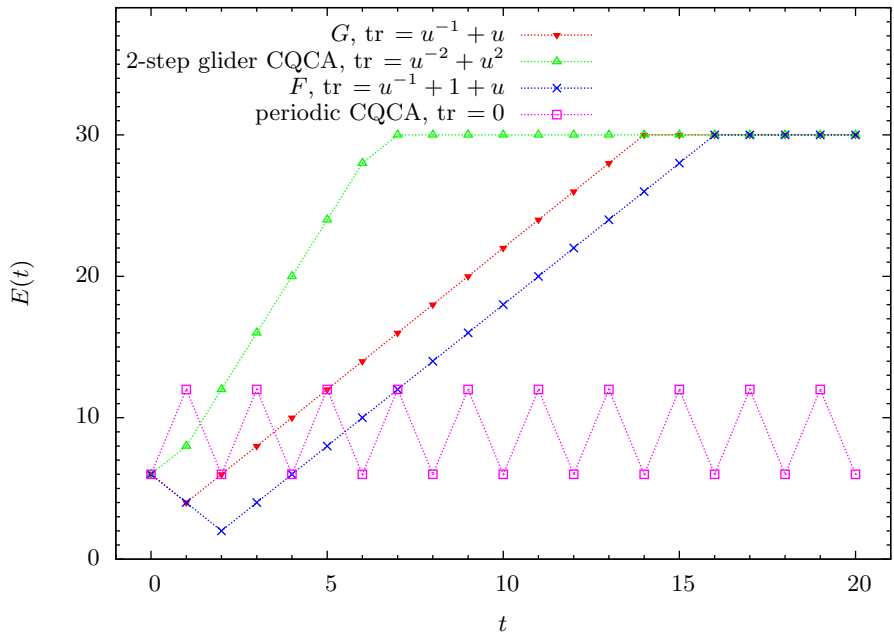}
    \caption{Evolution of entanglement of a subchain of $30$ consecutive spins for an initial  stabilizer state with $\weyl{\xi}=\sigma_2\otimes\sigma_1\otimes\sigma_1\otimes\sigma_1\otimes\sigma_1\otimes\sigma_1\otimes\sigma_2$ for different CQCA actions. The entanglement first grows as in the bipartite case, but then saturates at $30$ qubit pairs.}
    \label{fig:tripartite_stab}
  \end{center}
\end{figure}
%%%%%%%%%%%%
\subsection{Entanglement generation starting from translation-invariant quasifree states}
\label{sec:ent_qf}
In this section we study the entanglement
generation of the glider automorphism $\glider$ acting on a  
family of pure 
quasifree states that interpolates between the
all-spins-up state (discussed in the previous section)
and a glider-invariant state (discussed in Section \ref{sec:stat_qf}).

Let $\omega$ be a pure translation-invariant quasifree state,
and let $\omega_{[1,L]}$ denote its restriction to the lattice points
$\{1,2, \dots,L \}$.
The entanglement entropy of the restricted state $\omega_{[1,L]}$ can be calculated from the eigenvalues $\{\lambda_i \}_{i=1 \dots 2L}$ of the restricted majorana two-point matrix $\left[M_{n,m}\right]_{n,m=1}^{2L}$ 
by the formula \cite{Fannes,VLRK}:
\begin{equation}
S(\omega_{[1,L]}) = - \sum_{i=1}^{2L} \lambda_i/2 \log 
( \lambda_i/2). \label{bin} 
\end{equation}

The family of states that we will consider
as initial states are the pure translation-invariant
quasifree states $\omega_A$ described by the symbol (see Section \ref{sec:stat_qf})):
\begin{equation*}
Q^{(\omega_A)}(p)= 
\begin{bmatrix}
1-\chi_{[-\pi A, 0]}(p)+\chi_{[0,\pi A]}(p) & 
i[1-\chi_{[-\pi A,0]}(p)-\chi_{[0,\pi A]}(p)]\\
-i[1-\chi_{[-\pi A,0]}(p)-\chi_{[0,\pi A]}(p)] & 
1-\chi_{[-\pi A,0]}(p)+\chi_{[0,\pi A]}(p)
\end{bmatrix},
\end{equation*}
where $\chi_{[a,b]}$ denotes the characteristic function
of the interval $[a,b]$, 
and $A$ is some real number between $0$ and $1$.
The state corresponding to $A=0$ is the all-spins-up state,
while the state corresponding to $A=1$ is a glider-invariant state.
We have shown in Section \ref{sec:stat_qf}
that by applying the glider automorphism $n$-times
on $\omega_A$ one obtains a quasifree state $\omega_A^{(n)}$ 
belonging to the symbol
\begin{equation*}
Q^{(\omega_A^{(n)})}(p)= 
\begin{bmatrix}
1-\chi_{[-\pi A, 0]}(p)+\chi_{[0,\pi A]}(p) & 
i[1-\chi_{[-\pi A,0]}(p)-\chi_{[0,\pi A]}(p)]e^{2inp}\\
-i[1-\chi_{[-\pi A,0]}(p)-\chi_{[0,\pi A]}(p)]e^{-2inp} & 
1-\chi_{[-\pi A,0]}(p)+\chi_{[0,\pi A]}(p)
\end{bmatrix}.
\end{equation*}
Using this result and Formula (\ref{bin}), we calculated 
numerically the entanglement generation.
The results are shown in Figures \ref{fig:ent_gen_qf_a} and \ref{fig:ent_gen_qf_l}. 
We can observe that the entanglement generation is linear in time,
its rate is maximal when $A=0$,
and the rate can be arbitrarily small (when $A$ approaches $1$). This is illustrated in Figure \ref{fig:ent_gen_qf_a}. For longer subchains it takes more time steps for the entanglement to saturate. We show this in Figure \ref{fig:ent_gen_qf_l}.

%\begin{figure}[htbp]
\begin{figure}[ht]
  \begin{center}
    \input{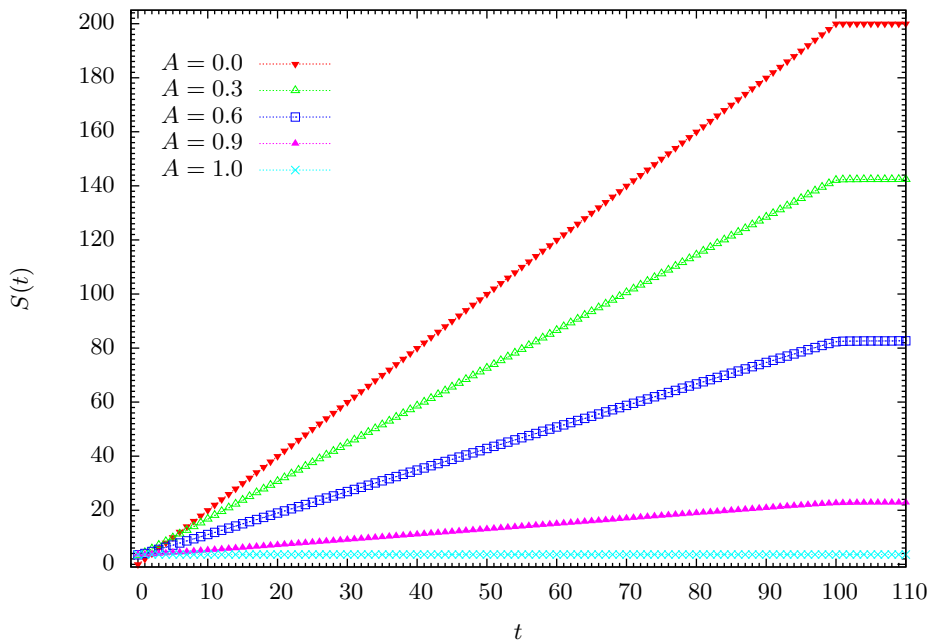}
    \caption{The entanglement entropy of a subchain of length $200$ after applying the glider time-evolution t times. Different colors mark the different initial of the parameter $A$ of the initial quasifree state.}
    \label{fig:ent_gen_qf_a}
  \end{center}
\end{figure}

%\begin{figure}[htbp]
\begin{figure}[ht]
  \begin{center}
    \input{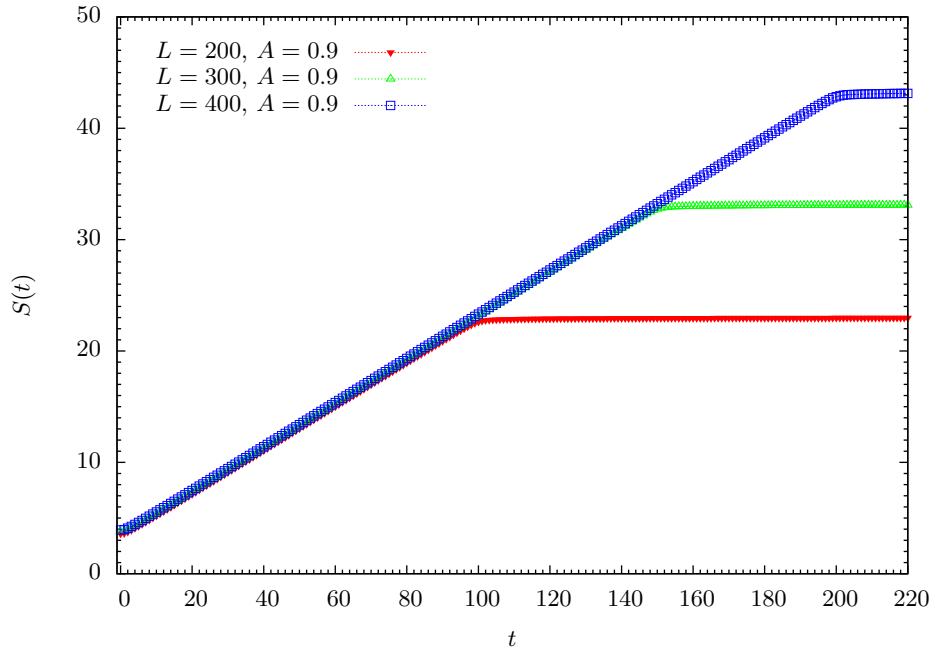}
    \caption{The entanglement entropy of a subchain of $L$ consecutive spins in dependence of the time steps. For larger $L$ more entanglement is possible. Different colors refer to different length of the chain. $A$ is fixed at $0.9$. }
    \label{fig:ent_gen_qf_l}
  \end{center}
\end{figure}
%%%%%%%%%%%%%%%%%%%%%%
%End of main content
%%%%%%%%%%%%%%%%%%%%%%%
\begin{acknowledgments}
  The authors would like to thank Vincent Nesme, Holger Vogts, Szil\'ard Farkas, and P\'eter Vecserny\'es for helpful discussions.
  Johannes G\"utschow is supported by the Rosa Luxemburg Foundation. The Braunschweig/Hannover group is supported by the DFG Forschergruppe 635, the European Union through the Integrated Project ``SCALA'' (grant number 015714) and the European FP6 STREP QICS project (grant number 033763). Zolt\'an Zimbor\'as is supported by the European project COQUIT under FET-Open grant number 233747.
\end{acknowledgments}
%%%
%%%
\begin{appendix}
%%%
%%%
\section{Proofs and technicalities}
  \subsection{Some results for fractal CQCA}
    We state here some results for fractal CQCAs. All results are stated for CQCAs, because it is more convenient than the pure phase space formulation. Nevertheless, the results are also true for CSCAs and the proofs use the phase space formulation. 
    \label{app:frac}
    \begin{lem}
      \label{lem:finite_pauli}
      A finite tensor product of only one kind of Pauli matrices (and the identity) occurs at most once for every Pauli matrix $\sigma_j,\,j=1,2,3$ in the history of any non-periodic CQCA $\CQCA$.
      \begin{proof}
        We begin with the case $j=1$. For two observables of the form $\bigotimes\sigma_i,\,i=1,0$ to occur in the same time evolution of a CQCA $\CQCA$, the condition
        \begin{equation*}
          \label{eq:cqca_same_Pauli}
          \sca^n=\left(
          \begin{array}{cc}
            x&y\\
            z&v
          \end{array}
          \right)
          \binom{a}{0}=\binom{b}{0}
        \end{equation*}
        has to be true. It immediately follows, that 
        \begin{equation*}
          \sca^n= \left(
          \begin{array}{rl}
              1 & y  \\
              0 & 1 \\
          \end{array} \right),
        \end{equation*}
        which is a periodic automaton. The case $j=3$ works analogous. For $j=2$ we employ the fact, that we can build a CQCA $\CQCA$ which fulfills (\ref{eq:cqca_same_Pauli}) via
        \begin{displaymath}
          \sca=\left(\begin{array}{rl}
              1 & 0  \\
              1 & 1 \\
          \end{array} \right)
          \scb
          \left(\begin{array}{rl}
              1 & 0  \\
              1 & 1 \\
          \end{array} \right)
        \end{displaymath}
        from any CQCA $B$, that fulfills $\scb\tbinom{a}{a}=\tbinom{b}{b}$. Moreover, as the conjugation with a CQCA does not change the trace, all CQCAs $B$ have to be periodic.
      \end{proof}
    \end{lem}

    \begin{lem}\label{Pauli_viele}
      Let $\CQCA$ be a fractal CQCA on a spin chain and let $\bigotimes \sigma_i$ be a finite tensor product of Pauli matrices. For every $k \in \NN$ there exist an $m \in \NN$ such that
      $\CQCA^m \bigotimes \sigma_i$ contains at least $k$ Pauli matrices.
      \begin{proof}
        We assume that the number of elements is bounded by some $k_{max}$. The area over which these $k_{max}$ elements are distributed is not bounded: If it were, the elements would either be restricted to a finite area for an infinite number of time steps, implying periodicity, or the area would move as a whole implying gliders for some power of $\CQCA$, which is not possible as shown in Lemma~\ref{lem:fractal_power}.
        So we see that any group of Pauli matrices will eventually be distributed over any area. But as we require the number of elements to be bounded, the distance between any two groups of Pauli matrices becomes larger than the neighborhood of the automaton. Then the starting argument applies to each of the new groups and forces them to break apart further until only isolated single-cell observables are left. But these will expand, thus the number of elements can't be bounded.
      \end{proof}
    \end{lem}
  
  \subsection{Proof of theorem~\ref{thm:stab_ent_encode}}\label{sec:app_stab_ent_encode}
    \begin{proof}[Proof of Theorem~\ref{thm:stab_ent_encode}]

    We use the criterion of Theorem~\ref{thm:stab_ent} and explicitly construct the pairs $\weyl{\xi_i}$, $\weyl{\eta_i}$ using methods from stabilizer codes for quantum error correction \cite{Gottesman1997}. As said in Section \ref{sec:ent_stab} only stabilizer generators localized on both sides of the cut are elements of the correlation group $\stab_{AB}$. The projections of all other stabilizer generators are just the stabilizer generators themselves which trivially commute with all other stabilizer generators and their projections on $A$ resp.\ $B$. We now use the phase space representation of the Pauli products and build the following $2n\times4n$-matrix from the cut translates of the stabilizer generators:
    \begin{equation*}
      \left(\begin{array}{cccc|cccc}
              \xi^{-n}_+&\xi^{-n+1}_+&\cdots&\xi^{n-1}_+&\xi^{-n}_-&\xi^{-n+1}_-&\cdots&\xi^{n-1}_-\\
              0&\xi^{-n}_+&\cdots&\xi^{n-2}_+&0&\xi^{-n}_-&\cdots&\xi^{n-2}_-\\
              \vdots&\ddots&\ddots&\vdots&\vdots&\ddots&\ddots&\vdots\\
              0&\cdots&0&\xi^{-n}_+&0&\cdots&0&\xi^{-n}_-
            \end{array}
      \right).
    \end{equation*}
    Let us assume that the outermost element is a $\sigma_1$\footnote{The other cases work equivalently.}. Then $\xi^{-n}_+=1$ and $\xi^{-n}_-=0$. From Section \ref{sec:ent_stab} we know, that at least one $\xi^{i}_-\ne0$. Let the $i$-th  diagonal of the right part be the first non-zero one. We also use the reflection invariance of $\xi$ to replace $\xi^{-j}$ by $\xi^j$ and get the following matrix: 
    \begin{equation*}
      \left(\begin{array}{ccccc|ccccccccccc}
              1&\xi^{n-1}_+&\xi^{n-2}_+&\cdots&\xi^{n-1}_+&0&\cdots&0&\xi^{n-i}_-&\xi^{n-i-1}_-&\cdots&\xi^{n-i}_-&0&\cdots&0\\
              0&1&\xi^{n-1}_+&\cdots&\xi^{n-2}_+&&&&&&&\ddots&&&\\
              &&&&&&&&&&0&\xi^{n-i}_-&\xi^{n-i-1}_-&\cdots&\xi^{n-i}_-\\
              \vdots&&\ddots&&\vdots&\vdots&&&&\ddots&&&&\ddots&\vdots\\
              &&&&&&&&&&&&&0&\xi^{n-i}_-\\
              &&&&&&&&&&&&&&0\\
              0&\cdots&0&1&\xi^{n-1}_+&&&&&&&&&&\vdots\\
              0&\cdots&0&0&1&0&&&&&\cdots&&&&0\\
            \end{array}
      \right).
    \end{equation*}
    Now we can perform the Gaussian algorithm on the matrix to obtain an identity matrix in the left part. As the rows are shifted copies of the first row, all operations will also be applied in a shifted copy. If we would add the third row to the first, we would also add the fourth to the second and so forth. We only add rows to rows above, because the lower left part of the matrix is already zero. We therefore get
    \begin{equation*}
      \left(\begin{array}{ccccc|ccccccccc}
              1&0&0&\cdots&0&0&\cdots&0&\zeta^{-n+i}&&\cdots&\zeta^{n-1}\\
              0&1&0&\cdots&0&&&&&&\ddots&\vdots\\
              &&&&&&&&&&0&\zeta^{-n+i}\\
              \vdots&&\ddots&&\vdots&\vdots&&&\ddots&&&0\\
              0&\cdots&0&1&0&&&&&&&\vdots\\
              0&\cdots&0&0&1&0&&&&\cdots&&0\\
            \end{array}
      \right).
    \end{equation*}
    The $i$-th row of the right part of the matrix remains unchanged, so $\zeta^{-n+i}=\xi^{n-i}_-=1$. We therefore get operators of the form
    \begin{equation*}
      s_k=\weyl{\tilde\xi_k}=\sigma_1^{-n}\otimes\left(\bigotimes_{j=-n+1}^{-n+i-1}\Id^{j}\right)\otimes\sigma_3^{-n+i}\otimes\left(\bigotimes_{j=-n+i+1}^{n-k}\sigma_?\right),\,1\le k\le 2n,
    \end{equation*}
    where the $\sigma_?$ can only be $\sigma_3$ or $\sigma_0=\Id$.
    We can easily see, that $\{s_k,s_{k+i-1}\}=0$. As $i\le n$ we can always find $n$ pairs of anticommuting operators. Unfortunately these pairs do not necessarily commute with other pairs. But through multiplication of operators we can find new pairs, which fulfill the necessary commutation relations. To show this we create a (symmetric) matrix $c_{i,j}=\sigma(s_i,s_j)$ of commutation relations.
    \begin{equation*}
      C=\begin{array}{r|ccccc}
        &s_1&s_i&s_j&s_{j+i-1}&\cdots\\
        \hline
        s_1&0&1&?&?&\cdots\\
        s_i&1&0&?&?&\cdots\\
        s_j&?&?&0&1&\cdots\\
        s_{j+i-1}&?&?&1&0&\cdots\\
        \vdots&\vdots&\vdots&\vdots&\vdots&\ddots
      \end{array}.
    \end{equation*}
    A ``$1$'' stands for anticommutation, a ``$0$'' for commutation. In the end we want all operators from different pairs to commute, so all positions denoted by question marks should get a zero entry. If we multiply operators, the corresponding rows and columns are added. Through these operations we can bring the commutation matrix to the form
    \begin{equation*}
      \widehat{C}=\begin{array}{r|ccccc}
        &s_1&s_i&s_j&s_{j+i-1}&\cdots\\
        \hline
        s_1&0&1&0&0&\cdots\\
        s_i&1&0&0&0&\cdots\\
        s_j&0&0&0&1&\cdots\\
        s_{j+i-1}&0&0&1&0&\cdots\\
        \vdots&\vdots&\vdots&\vdots&\vdots&\ddots
      \end{array}.
    \end{equation*}
    To show that this is possible, we will consider a prototype of such an operation. Given the matrix 
    \begin{equation*}
      C=\begin{array}{r|ccccc}
        &s_1&s_i&s_j&s_{j+i-1}&\cdots\\
        \hline
        s_1&0&1&c_{13}&c_{14}&\cdots\\
        s_i&1&0&c_{23}&c_{24}&\cdots\\
        s_j&c_{13}&c_{23}&0&1&\cdots\\
        s_{j+i-1}&c_{14}&c_{24}&1&0&\cdots\\
        \vdots&\vdots&\vdots&\vdots&\vdots&\ddots
      \end{array}.
    \end{equation*}
    We now pick a nonzero $c_{ij}$ and do the following. If $c_{ij}=1$ and $i$ odd, we add row $i+1$ to row $j$ and the same for the columns. If $c_{ij}=1$ and $i$ even, we add row $i-1$ to row $j$ and the same for the columns. This only changes the one $c_{ij}$ we are considering, the others remain unchanged. After each step we get a new matrix $\widetilde{C}$ and pick another nonzero $c_{ij}$ from the same $2\times 2$ block (in this example we only have one block). By doing this for all blocks in the first two rows, we create pairs of operators that commute with the first pair. Now we have to check if the process destroyed the anticommutation within the pairs. The diagonal entries of the matrix trivially stay zero, because all operators commute with themselves. We only have to check the other elements of the block (due to the symmetry, we only have to check one). So if $c_{13}=1$ we get $1\mapsto 1+c_{24}$. We can write $1\mapsto 1+c_{13}c_{24}$. Including the whole block of $c_{ij}$ we get $1\mapsto 1+c_{13}c_{24}+c_{14}c_{23}+c_{23}c_{14}+c_{24}c_{13}=1$ as all operations are carried out modulo $2$. The new pairs thus fulfill the anticommutation condition. We can repeat this process for the new pairs until all operators from different pairs commute. We started with $2n$ operators, thus we arrived at $n$ pairs which together with their counterparts on the other subsystem encode $n$ pairs of maximally entangled qubits.
  \end{proof}
    %%

%%%
%%%
\end{appendix}
%%%
%%%
%%%
\bibliography{cqca_assym_paper}
%%%
%%%
\end{document}